\documentclass[a4paper,autoref]{article}
\usepackage[utf8]{inputenc}
\usepackage{xspace}
\usepackage{amsthm}
\usepackage{amsmath}
\usepackage{tabularx,booktabs}
\usepackage{thmtools, thm-restate}
\usepackage{hyperref}
 \usepackage[a4paper,%
             left=1in,right=1in,top=1in,bottom=1in,%
             ]{geometry}
\usepackage{cleveref}
\usepackage{cite}
\usepackage{graphicx}
\usepackage[normalem]{ulem}
\usepackage[charter]{mathdesign}
\usepackage[scale=0.95]{sourcesanspro}	
\usepackage{inconsolata}
\usepackage{arydshln,cellspace}
\renewcommand{\paragraph}[1]{\vspace{7pt} \noindent \textbf{#1}}
\newcommand{\R}{\ensuremath{\mathbb R}\xspace}
\newcommand{\etal}{{et al.}\vspace{-1pt}\xspace}

\newcommand{\Q}{\ensuremath{\mathbb Q}\xspace}
\newcommand{\QQ}{\ensuremath{Q}\xspace}

\newcommand{\Z}{\ensuremath{\mathbb Z}\xspace}
\newcommand{\N}{\ensuremath{\mathbb N}\xspace}
\newcommand{\eps}{\ensuremath{\varepsilon}\xspace}
\newcommand {\E}{\ensuremath{\mathbb{E}}}
\newcommand{\ER}{\ensuremath{\exists \mathbb{R}}\xspace}
\newcommand{\ETR}{\ensuremath{\textrm{ETR}}\xspace}
\newcommand{\Opt}{\ensuremath{\textrm{opt}}\xspace}
\newcommand{\NP}{\ensuremath{\textrm{NP}}\xspace}
\newcommand{\PSPACE}{\ensuremath{\textrm{PSPACE}}\xspace}
\newcommand {\SmoothAna}{smoothed analysis\xspace}
\newcommand {\realRam}{real RAM\xspace}
\newcommand {\wordRam}{word RAM\xspace}
\newcommand {\VoronoiDiagram}{Voronoi diagram\xspace}
\newcommand {\voronoicell}{Voronoi cell\xspace}
\newcommand {\realVerify}{real verification algorithm\xspace}

\newcommand {\gvAlgo}{recognition verification algorithm\xspace}
\newcommand {\mono}{monotonous\xspace}
\newcommand{\breakproperty}{moderate\xspace}
\newcommand{\bitaugmented}{smoothable\xspace}
\newcommand{\wordsize}{w\xspace}
\newcommand{\width}{\ensuremath{\omega}\xspace}

\renewcommand{\line}{{\ensuremath{\textrm{segment}}}\xspace}

\newcommand{\Inst}{I}
\newcommand{\Rcert}{x}
\newcommand{\Zcert}{z}
\newcommand{\IfThenElse}[3]{\textsf{if~~} #1 \textsf{~~then~~} #2 \textsf{~~else~~} #3}

\newcommand {\InputBitComp}{input-precision\xspace}
\newcommand {\BitComp}{bit-precision\xspace}
\newcommand {\BitLength}{bit-length\xspace}
\newcommand {\AlgDim}{algebraic dimension\xspace}
\newcommand {\AlgDeg}{arithmetic degree\xspace}
\newcommand {\fewPolynomials}{few polynomials\xspace}
\newcommand {\totalPoly}{total number of polynomials\xspace}
\newcommand {\TotalPoly}{Total number of polynomials\xspace}

\newcommand {\BIT}{\texttt{bit}}
\newcommand {\BitInput}{\texttt{bit}_\texttt{IN}}

\newcommand{\var}[1]{\ensuremath{\llbracket #1 \rrbracket}\xspace}

\newcommand{\floor}[1]{\lfloor #1 \rfloor}

\newcommand{\sign}{\ensuremath{\textnormal{sign}}}		

\newcommand{\range}[2]{#1\,..\,#2}

 \newtheorem{theorem}{Theorem}
 \newtheorem{corollary}[theorem]{Corollary}
 \newtheorem{lemma}[theorem]{Lemma}

\title{
Smoothing the gap between NP and ER
}
\author{Jeff Erickson, University of Illinois, jeffe@illinois.edu\\
Ivor van der Hoog, Utrecht University, i.d.vanderhoog@uu.nl\\
Tillmann Miltzow, Utrecht University, t.mitzow@uu.nl}

\begin{document}

\maketitle
\thispagestyle{empty}

\begin{abstract}
We study algorithmic problems that belong to the complexity class of the existential theory of the reals (\ER). A problem is \ER-complete if it is as hard as the problem ETR and if it can be written as an ETR formula.
Traditionally, these problems are studied in the \realRam, a model of computation that assumes that the storage and comparison of real-valued numbers can be done in constant space and time, with infinite precision.
 The complexity class \ER is often called a \realRam analogue of \NP, since the problem ETR can be viewed as the real-valued variant of SAT.
The \realRam assumption that we can represent and compare arbitrary irrational values in constant space and time is not very realistic. 
Yet this assumption is vital, since some \ER-complete problems have an ``exponential bit phenomenon'' where there exists an input for the problem, such that the witness of the solution requires geometric coordinates which need exponential word size when represented in binary. 
The problems that exhibit this phenomenon are NP-hard (since ETR is NP-hard) but it is unknown if they lie in \NP. \NP membership is often showed by using the famous Cook-Levin theorem which states that the existence of a polynomial-time verification algorithm for the problem witness is equivalent to \NP membership. The exponential bit phenomenon prohibits a straightforward application of the Cook-Levin theorem.

In this paper we first present a result which we believe to be of independent interest: we prove a \realRam analogue to the Cook-Levin theorem which shows that \ER membership is equivalent to having a verification algorithm that runs in polynomial-time on a \realRam. This gives an easy proof of \ER-membership, as
verification algorithms on a \realRam are 
much more versatile than \ETR-formulas.

We use this result to construct a framework to study \ER-complete problems under \SmoothAna.
We show that for a wide class of \ER-complete problems, its witness can be represented with logarithmic \InputBitComp by using \SmoothAna on its \realRam verification algorithm. This shows in a formal way that the boundary between \NP and \ER (formed by inputs whose solution witness needs high \InputBitComp) consists of contrived input.

We apply our framework to well-studied \ER-complete recognition problems which have the exponential bit phenomenon such as the recognition of realizable order types or the Steinitz problem in fixed dimension. Interestingly our techniques also generalize to problems with a natural notion of resource augmentation (geometric packing, the art gallery problem). 

A prior version of this paper appeared in FOCS 2020~\cite{FOCSRobustComputation}.
\end{abstract}

\hfill

\paragraph{Acknowledgements.}
The second author is supported by the Netherlands Organisation for Scientific Research (NWO); 614.001.504.
The third author is supported by the Netherlands Organisation for Scientific Research (NWO) under project no. \textrm{016.Veni.192.25}.
We thank anonymous reviewers for valuable and detailed comments.
\clearpage
\pagenumbering{arabic} 

\section{Introduction}

The RAM is a mathematical model of a computer 
which emulates how 
a computer can access and manipulate data. 
Within computational geometry,
algorithms are often analyzed within the \realRam
\cite{salesin1989epsilon,fortune1993efficient, li2005recent} 
where real values can be stored and 
compared in constant space 
and time. By allowing these infinite precision computations, it 
becomes possible to verify 
geometric primitives in constant 
time, which simplifies the analysis of geometric algorithms. 
Mairson and Stolfi~\cite{mairson1988reporting} 
point out that
``without this assumption it is
virtually impossible to prove the correctness 
of any geometric algorithms.''

Intuitively (for a formal definition, skip ahead to the bottom of page 2) we define  the \InputBitComp of a \realRam algorithm $A$ with input $I$ as the minimal word size required to express each input value in $I$, such that the algorithm executes the same operations in the \wordRam as in the \realRam.
The downside of algorithm analysis in the \realRam is that 
it neglects the \InputBitComp required by the underlying
algorithms for correct execution, although they are very important in practice.
Many algorithms, including some verification algorithms of $\ER$-complete problems, inherently require large \InputBitComp~\cite{salesin1989epsilon} and there are even examples which in the worst case require an \InputBitComp exponential in the number of input variables in order to be correctly executed~\cite{kammer2014practical, ExpCoord}.

Often inputs which
require exponential \InputBitComp are contrived and 
do not resemble \emph{realistic} inputs. 
A natural way to capture this from a theoretical 
perspective is \SmoothAna, which \emph{smoothly} interpolates 
 between
worst case analysis and average case analysis~\cite{spielman2004smoothed}.
Practical inputs are constructed inherently
with small amount of noise and random perturbation. 
This perturbation helps to show performance guarantees
in terms of the input size and the magnitude of the perturbation.
By now \SmoothAna is well-established, for instance 
Spielman and Teng received the G\"{o}del Prize for it. 

We give a bird's eye view of the paper before providing proper definitions:  we study the \realRam and its complexity class \ER through the lens of \SmoothAna. 
We start by proving a result separate of \SmoothAna, which we believe to be of independent interest: we show a \realRam analogue of the famous Cook-Levin theorem as we prove that \ER membership is equivalent to the existence of a verification algorithm that runs in polynomial-time on a \realRam.
We then use the existence of this verification algorithm to construct a framework with which \SmoothAna can be applied to a wide class of \ER-complete problems. We show that \ER-complete recognition problems, with a verification algorithm in the \realRam that has polynomially bounded \emph{\AlgDeg}, have witnesses that can be represented with a logarithmic word size under \SmoothAna. 
This implies that these problems have polynomial-time verification algorithms on the \wordRam that succeed for all but a small set of contrived inputs.
Finally, we extend our framework to include \ER-complete problems that have a natural notion of resource augmentation. This generalizes an earlier result of \SmoothAna of the Art Gallery problem~\cite{ArxivSmoothedART}.

\paragraph{RAM computations.}
The Random Access Machine (RAM) 
is a model of computation for the standard computer architecture. 
The precise definition of the RAM varies,
but at its core the RAM has a 
number of \emph{registers}
and a \emph{central processing unit} (CPU),
which can perform operations on register values 
like reading, writing, comparisons, and arithmetic operations.
The canonical model of computation within
computer science 
is the \emph{\wordRam}, 
a variation on the RAM formalized by Hagerup~\cite{h-sswr-98} but 
previously considered by Fredman and Willard \cite{fw-sitbf-93,fw-tams-94}
and even earlier by Kirkpatrick and Reisch \cite{kr-ubsir-84}.
The \wordRam 
models two crucial
aspects of real-life computing: (1) computers must store values with 
finite precision and (2) computers take more time to perform 
computations if the input of the computation is longer.
Specifically, the word RAM supports constant-time operations on
$\wordsize$-bit integers, where the \emph{word size} $\wordsize$ is a parameter of the model.

Many portions of the algorithms community (either explicitly or implicitly) use a different variation of the RAM called the \emph{\realRam},
where registers may contain arbitrary real numbers, instead of just integers;
The usage of the \realRam is prevalent in the field of computational geometry but also in probabilistic algorithm analysis where one wants to reason about continuous perturbations of the input.
The abstraction offered by the \realRam dramatically simplifies the design and analysis of algorithms,  at the cost of working in a physically unrealistic model.
Implementations of \realRam algorithms using finite-precision data types are prone to errors,
not only because the output becomes imprecise, but because rounding errors can lead the algorithm into inconsistent states. Kettner~\cite{kettner2008classroom} provides an overview of complications that arise from the unrealistic precision that the \realRam assumes.

\begin{figure}[t]
    \centering
    \includegraphics{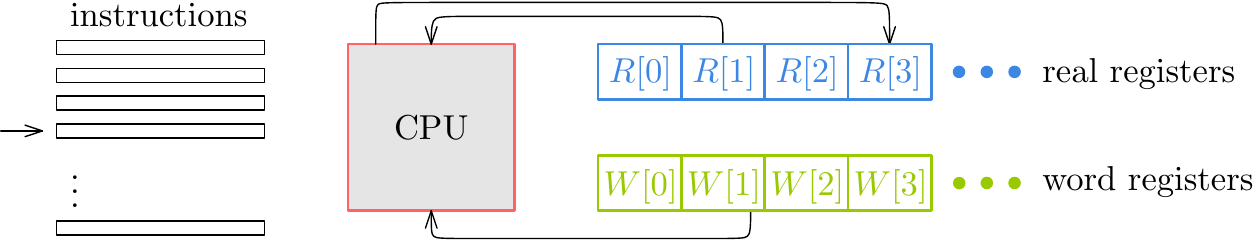}
    \caption{The dominant model in computational geometry is the \realRam. 
    It consists of a central processing unit, which can operate
    on real and word registers in constant time, following a set of instructions.}
    \label{fig:RAM}
\end{figure}

\paragraph{Formally modeling \realRam algorithms.}
The \realRam has been the standard underlying model of computation in computational geometry since the field was founded in the late 1970s \cite{s-cg-79, ps-cgi-85}.  
Despite its ubiquity, we are unaware of \emph{any} published definition of the model that is simultaneously precise enough to support our results and broad enough to encompass most algorithms in the algorithm analysis literature.  The obvious candidate for such a definition is the real-computation model proposed by Blum, Shub, and Smale \cite{blum1989theory, bcss-crc-98}; however, this model does not support the integer operations necessary to implement even simple algorithms:
even though the \realRam is often presented, either formally or intuitively,
as a random access machine that stores and manipulates only exact real numbers,
countless algorithms in this model require decisions based on both exact real
and finite precision integer values.  Consider the following
example: given an array of $n$ real values as input, compute their sum.
Any algorithm that computes this sum must store and manipulate real numbers;
however, the most straightforward algorithm also requires indirect memory access 
through an \emph{integer} array index.
More complex examples include call stack maintenance, discrete symbol manipulation, and
multidimensional array indexing and program slicing.

On the other hand, real and integer operations must be combined with care to avoid  unreasonable \emph{discrete} computation power.  A model that supports both exact constant-time real arithmetic and constant-time conversion between real numbers and integers, for example using the floor function, would also trivially support arbitrary-precision constant-time \emph{integer} arithmetic.  (To multiply two integers, cast them both to reals, multiply them, and cast the result back to an integer.)  Including such constant-time operations allows any problem in PSPACE to be solved in polynomial-time \cite{s-pram-79}; see also~\cite{kr-ubsir-84, hs-pmram-74, bms-scram-85, e-mms-90} for similar results.

To accommodate this mixture of real and integer operations, and to avoid complexity pitfalls, we define the 
\realRam 
as an extension of the standard integer \wordRam 
\cite{h-sswr-98} 
(refer to \Cref{fig:RAM}). 
We define the \realRam in terms of a fixed parameter $w$, called the \emph{word size}.  A \emph{word} is an integer between $0$ and $2^w-1$, represented as a sequence of $w$ bits.
 Mirroring standard definitions for the \wordRam, memory consists of two \emph{random access arrays} $W[\range{0}{2^w-1}]$ and $R[\range{0}{2^w-1}]$, whose elements we call \emph{registers}.  Both of these arrays are indexed/addressed by words; for any word~$i$, register $W[i]$ is a word and register $R[i]$ is an exact real number.  (We sometimes refer to a word as an \emph{address} when it is used as an index into a memory array.)

A program on the \realRam consists of a fixed, finite indexed sequence of read-only instructions.  The machine maintains an integer \emph{program counter}, which is initially equal to $1$.  At each time step, the machine executes the instruction indicated by the program counter.  The \textsf{goto} instruction modifies the program counter directly; the \textsf{halt} and \textsf{accept} and \textsf{reject} instructions halt execution
otherwise, the program counter increases by~$1$ after each instruction is executed.

The input to a \realRam program consists of a pair of vectors $(a,b) \in \R^n \times \Z^m$, for some integers $n$ and $m$, which are suitably encoded into the corresponding memory arrays before the program begins.\footnote{Following standard practice, we implicitly assume throughout the paper that the integers in the input vector~$b$ are actually $w$-bit \emph{words}; for problems involving larger integers, we take $m$ to be the number of \emph{words} required to encode the integer part of the input.}  To maintain uniformity, we require that neither the input sizes $n$ and $m$ nor the word size $w$ is known to any program at “compile time”.  The output of a \realRam program consists of the contents of memory when the program executes the \textsf{halt} instruction.
The \emph{running time} of a \realRam program is the number of instructions executed before the program halts; each instruction requires one time step by definition.
Refer to \Cref{sec:DefRealRam} for a table of the operations and further explanation of the \realRam. 
 Crucially, we consider the \realRam without trigonometric, exponential and logarithmic operations. 
 We do allow for the use of the square root operator for our results regarding \ER-completeness but not for our result on smoothed analysis.
 
\paragraph{A recent addition to the RAM models.}
Having detailed our proposal of RAM model, we wish to briefly mention a recent proposal made by Demaine, Hesterberg and Ku~\cite{demaine2020finding}.
In this recent paper, the authors note similarly that the word RAM does not support operations with the appropriate degree of precision for algorithms in computational geometry. Our definition of the real RAM offers a way to theoretically model the theoretical computational power needed for constant-time real-valued computations and we obtain this model by extending the word RAM model.
Demaine, Hesterberg and Ku propose their own extension of the word RAM model, which is theoretically less general but which has a more direct relation to real-life computation power.

The authors note that if a constant-size radical expression (this includes the addition, multiplication and division of square roots) depends on variables of $b$ bits, then an algorithm can decide if the expression is greater than $0$ in $O(b)$ time. 
Following this observation, the authors assume a word RAM, which allows memory cells to not only store expressions of values with wordsize $b$, but these constant-size radical expressions. This makes their model more general than a traditional word RAM. This generalization is useful, as for example geometric operations or algorithms that require the computation of Euclidean distance are hereby supported. 

We want to mention that their model of computation is not directly applicable to our results for smoothed analysis or \ER-completeness, as both problems per definition require a way to model constant-time arithmetic on arbitrary reals. 
However we found it important to note that there are recent discussions on how and when radical expressions can be allowed in RAM computations.

\paragraph{Formally defining \BitComp and \InputBitComp.}
We say two inputs  
$I = (a,b) \in \R^n\times \Z^m$
and 
$I'= (a',b) \in \R^n\times \Z^m$
are equivalent with respect to an algorithm $A$ on a 
\realRam,
$(I \cong_A I')$ if 
every comparison operation gives the same result.
In particular, this implies that at every step the 
program counter is at the same position.
For each integer $z \in \Z$, we denote by $\BIT(z)$ the length of its binary representation.
For each rational number $y= p/q \in \Q$, with $p,q$ irreducible, we denote by $\BIT(y) := \BIT(p) + \BIT(q)$ the length of its binary representation.

First we consider $I = (a,b) \in \Q^n \times \Z^m$ as input for a \realRam algorithm $A$. 
We denote by $\BitInput(I) =
\max_i \; \max \{\BIT(a_i), \BIT(b_i) \}$ the \emph{input \BitLength} of $I$.
We denote by $C(I)$ the set of all values of all registers during the execution of $A$ with $I$, and by $\BIT(C(I)) = \underset{c \in C(I)}{\max} \{ \BIT(c)\}$ as the \emph{ execution \BitLength} of $I$.

Now, we are ready to define the \BitComp and \InputBitComp of \emph{real} input.
The \BitComp of $A$ with input $I = (a,b)\in\R^n\times \Z^m$ is: 
\[
\BIT(I, A) := {\min} \{\,  \BIT(C(I')) \mid I' \in \Q^n \times \Z^m,  I \cong_A I' \, \} .
\]
In the same manner, the  \InputBitComp 
is defined as :
\[
\BitInput(I, A) := {\min} \{\,  \BitInput(I') \mid I' \in \Q^n \times \Z^m,  I \cong_A I' \, \} .
\]
If there is no equivalent input $I' = (a,b)\in \Q^n \times \Z^m$, then we say $\BIT(I,A) = \BitInput(I, A) = \infty$.
It is now straightforward to \emph{simulate} an execution 
of a \realRam algorithm $A$ on input 
$I = (a,b) \in \R^n \times \Z^m$
on a \wordRam with word size $w = O(\BIT(I,A))$.

\paragraph{Arithmetic degree.}
Liotta, Preparata and Tamassia~\cite{liotta1998robust} studied the required \BitComp for \realRam proximity algorithms. To aid their analysis, they defined the \emph{arithmetic degree} of an algorithm. We express their definition in our \realRam model and add the notion of \AlgDim for our later analysis.
It follows from our definition of \realRam and our list of operations in \Cref{tab:instructions} that at all times during the computation, a 
real register holds a value which can be described
as the quotient of two polynomials $\frac{p}{q}$ 
of the real input values $a$. Observe that adding, subtracting, or multiplying two rational functions yields another rational function, possibly of higher degree; for example, $\tfrac{p_1}{ q_1} + \tfrac{p_2}{ q_2} = \tfrac{p_1q_2 + p_2 q_1}{q_1q_2}$.
We say an algorithm~$A$ has \emph{\AlgDeg}
$\Delta$, if $p$ and $q$ always have
total degree at most $\Delta$.
Similarly, $A$ has \emph{\AlgDim} $d$,
if the  number of variables in $p$ and $q$ is always 
at most $d$.
Finally, we define the \emph{max (integer) coefficient} $c$, if the
absolute value of the largest coefficient of $p$ and $q$ is always bounded from above by $c$.
Bounded \AlgDim, \AlgDeg and max coefficient give an interesting relation between the \InputBitComp and the \BitComp:

\begin{lemma}
\label{lem:intermediate}
    Let $I\in (a,b) \in \R^n \times \Z^m$ be some input and $A$ be a \realRam algorithm with $\BitInput(I, A) = p$, \AlgDim $d$, \AlgDeg $\Delta$ and max coefficient $c$. Then its \BitComp is bounded  from above  by
    $O(p \Delta^2 \log d \log c)$.
\end{lemma}
\begin{proof}
If we multiply a number of \BitLength $b$ with a number $c$ then the resulting number has bit complexity at most $O(b \log c)$.
If we multiply $\Delta$ numbers of \BitLength $p$, the total \BitLength is at most $\Delta p$.
If we sum $t$ numbers of \BitLength $p$, the total \BitLength is at most $O(p \log t)$.
In a polynomial in $d$ variables, we have at most $O(d^{\Delta})$ monomials.
Thus the \BitComp is bounded  from above by 
$O(\log c \log( d^\Delta) \Delta p) = O(p \Delta^2 \log d \log c)$.
\end{proof}

This allows us to only focus on \InputBitComp in our \SmoothAna which we explain next:  

\paragraph{Smoothed analysis.}
In \emph{\SmoothAna{}}, the performance of an algorithm 
is studied for worst case input which is 
randomly perturbed 
by a magnitude of~$\delta$. (that is, we consider out of all inputs the input that is least-beneficial to perturb from) 
Intuitively, \SmoothAna interpolates between average case 
and worst case analysis (\Cref{fig:ConceptSmoothedAnalysis}). 
The smaller~$\delta$, the closer
we are to true worst case input.
Correspondingly larger~$\delta$ is closer to 
the average case analysis. The key difficulty in applying \SmoothAna
is that one has to argue about both worst case 
and average case input.
Spielman and Teng explain their analysis by applying it to 
the simplex algorithm,
which was known for a particularly good performance in practice 
that was seemingly impossible to 
verify theoretically~\cite{klee1970good}.
Since the
introduction of \SmoothAna, 
it has been applied to numerous problems~\cite{nemhauser1969discrete,KnapsackSmooth, arthur2006worst, englert2007worst,MaxCUTsmoothed}.
Most relevant for us is the \SmoothAna of the 
art gallery problem~\cite{ArxivSmoothedART} 
and of order types~\cite{van2019smoothed}, which we generalize  in \Cref{sub:resourceAugmentation}.

\begin{figure}[t]
    \centering
    \includegraphics{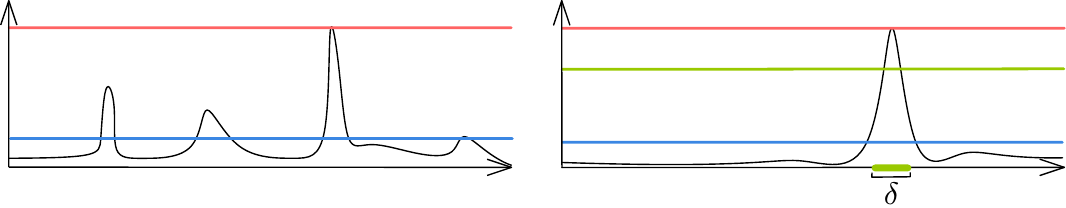}
    \caption{ The $x$-axis represents all inputs for the algorithmic problem. The $y$-axis the associated cost for every input. The red line indicates the worst case cost. The blue line average cost. If we pick a parameter $\delta$ and a point $a$ on the $x$-axis, then all permutations from $a$ of at most $\delta$ form some area centered around $a$. We can compute the average cost of points in this area. The cost under \SmoothAna is the maximum over all points $a$, of this average cost and is shown by the green line.}
    \label{fig:ConceptSmoothedAnalysis}
\end{figure}

Formally we define \SmoothAna as follows: let us fix an algorithm $A$ and some $\delta \in [0,1]$, 
which describes the \emph{magnitude of perturbation}.
We denote by $I= (a,b)\in [0,1]^{n}\times \Z^m$ the input of $A$.
We scale the real-valued input to lie in $[0,1]$,
to normalize $\delta$.
Unless explicitly stated, we assume that each real number is perturbed \emph{independently uniformly at random} (each real number receives an offset sampled uniformly at random from $[\frac{-\delta}{2}, \frac{\delta}{2}]$). We assume that the integer input stays as it is,  since we assume that it already fits into main memory.
We denote by ($\Omega_\delta$, $\mu_\delta$) 
the probability space where each $x \in \Omega_\delta$
defines for each instance $I$ a new `perturbed' instance $I_x = (a+x,b)$. 
We denote by $\mathcal{C}(\cdot)$ an arbitrary cost function and by $\mathcal{C}(I_x)$ 
the cost of instance $I_x$. 
Traditionally in \SmoothAna, this cost is the runtime required for an algorithm in order to compute its solution but in this paper we consider the \InputBitComp and the \BitComp as the cost functions. 
The expected cost of instance $I$ equals:
\[ \mathcal{C}_\delta(I, A) = \underset{x\in\Omega_\delta}{\E} \left[ \mathcal{C}(I_x) \right] = \int_{\Omega_\delta} \mathcal{C}(I_x)\mu_\delta(x) \; \mathrm{d}x. \]
We denote by $\Lambda_{n, m}$ the set of all instances in $[0, 1]^n \times \Z^m$. Henceforth, we implicitly assume that for all integer values $b_i$, $\BIT(b_i) \leq \log \; m$ and $m = O(n)$. We drop the $m$ from all future equations and consideration. 
The smoothed \InputBitComp equals:
\[ \mathcal{C}_{\textrm{smooth}}(\delta, n, A) =
 \max_{I\in\Lambda_{n,m}} \mathcal{C}_\delta(I, A)   
.\]

This definition formalizes the intuition mentioned before: 
not only do we require that the majority of instances behave nicely, 
but actually in every neighborhood 
(bounded by the maximal perturbation $\delta$) the majority of 
instances must behave nicely. Following~\cite{dadush2018friendly, spielman2004smoothed} 
we say an algorithm runs in polynomial cost in practice,
if the smoothed cost of the algorithm is polynomial in 
the input size $n$ and in $1 / \delta$. If the smoothed cost is small in 
terms of $1/\delta$ then we have 
a theoretical verification of the hypothesis 
that worst case examples are sparse.
We defined the \InputBitComp of an algorithm $A$ with input $I$ as the minimum word size required for the correct execution of the algorithm on a \wordRam. To model possible disparities between \realRam and \wordRam execution we define an operation called \emph{snapping} in the next paragraph.

\paragraph{Snapping.}
Our real-valued input can 
be represented as a higher-dimensional point $a \in [0,1]^n$. 
If we want to express $a$ with only $w$ bits, then  the corresponding integer-valued input $a'$ with 
limited precision is the closest point to $a$ in the scaled integer lattice 
$\Gamma_\width = \width \Z^d$ with $\width =  2^{-w}$. We call the transformation of $a$ into $a'$ \emph{snapping}. 
We give a lower bound on the scale factor $\width$ for which $(a,b) \cong_A (a',b)$. 
From this point on, whenever we refer to bounding  from above the required \InputBitComp, we refer to a lower bound on the scale factor $\width$ such that $(a,b) \cong_A (a',b)$.
Note that this lower bound, implies an bound  from above  on the \InputBitComp of $A$. That is, $\width$ implies a  bound  from above on the number of bits needed to simulate the execution of an algorithm $A$ with input $(a, b)$ on the \wordRam. Naturally, it could be that there is a better way to represent or round $(a, b)$ to a value $(a'', b)$ than our snapping procedure. However, we show that our upper bound is logarithmic in the size of the real input $n$. A natural lower bound for the \InputBitComp is $\Omega(\log n)$ and thus our future results can be considered tight.
The algorithms that we study under 
this snapping operation are algorithms relevant 
to the complexity class $\ER$, which we discuss next.

\paragraph{The Existential Theory of the Reals.}
It is often sometimes to describe a real-valued witness to an \NP-hard
problem, but the \InputBitComp required to verify the witness is unknown.
A concrete example is the recognition of segment
intersection graphs: given a graph, can we represent
it as the intersection graph of segments?
Matou{\v{s}}ek~\cite{matousek2014intersection} 
comments on this as follows:  
\begin{quote}
\hspace{-15pt}{\small\emph{{Serious people seriously
conjectured that the number of digits can be polynomially bounded—but it cannot.}} }
\end{quote}
Indeed, there are examples which require an exponential word size 
in any numerical representation.
This \emph{exponential bit phenomenon} occurs not only for segment intersection
graphs, but also for many other natural algorithmic problems~\cite{abel, ARTETR, etrPacking, bienstock1991some, cardinal2017intersection, cardinal2017recognition, AreasKleist, NestedPolytopesER, erickson2019optimal, garg2015etr,  kang2011sphere, LindaPHD, AnnaPreparation, mcdiarmid2013integer,  mnev1988universality, richter1995realization,
Schaefer2010, schaefer2013realizability, Schaefer-ETR,   shitov2016universality, Shitov16a, shor1991stretchability, TrainNN, ContinuousCSP}.
It turns out that all of those algorithmic problems do not 
accidentally require exponential \InputBitComp, but are closely linked,
as they are all complete for a certain complexity class called \ER.
Thus either all of those problems belong to \NP, or none of them do.
Using our results on \SmoothAna, we show 
that for many \ER-hard problems
the exponential \InputBitComp phenomenon only occurs 
for near-degenerate input.

The complexity class \ER can be defined as the set of decision
problems that can be reduced in polynomial time on a \wordRam to deciding if
a formula of the 
\emph{Existential Theory of the Reals} (\ETR) is true or not.
An \ETR formula has the form: $\Psi = \exists x_1,\ldots,x_n \quad \Phi(x_1,\ldots,x_n),$ where $\Phi$ is a well-formed sentence over the alphabet $\Sigma =\{0,1,x_1,\ldots, +, \cdot, =, \leq , < , \land
, \lor, \lnot\}.$
More specifically, $\Phi$ is quantifier-free and 
$x_1,\ldots,x_n$ are all variables of $\Phi$.
We say the ETR formula $\Psi$ is true if there are
real numbers $x_1,\ldots,x_n\in \R$ such that
$\Phi(x_1,\ldots,x_n)$ is true.

\begin{figure}
    \centering
    \includegraphics{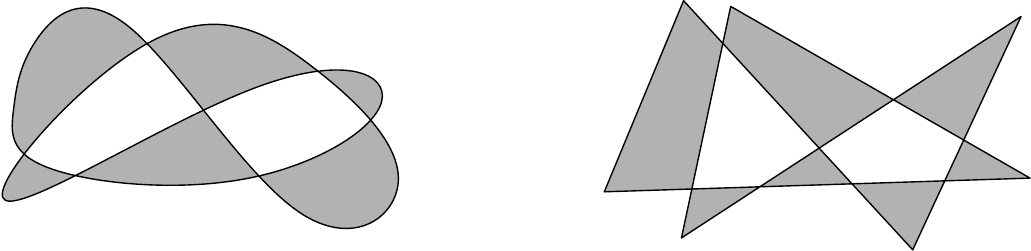}
    \caption{Given two simple closed curves in 
    the plane it is straightforward to design an algorithm 
    that checks if the two curves are equivalent. But it is 
    not straightforward to describe an \ETR formula for it. }
    \label{fig:CurveStraigthening}
\end{figure}

\paragraph{Paper structure.}
In \Cref{sec:membership}, we give a structural result that gives an alternative definition of \ER in terms of \realRam algorithms.
The corresponding proofs can be found in \Cref{sec:Cook-Levin}.
In \Cref{sec:Smoothed-Recognition}, we apply \SmoothAna, first to the \realRam and then to recognition problems.
The corresponding proofs can be found in \Cref{sec:HittingCubes} and \Cref{sec:Smoothing}.
In \Cref{sub:resourceAugmentation}, we apply \SmoothAna to resource augmentation problems.
The corresponding proofs can be found in \Cref{sec:ResourceAugmentation}.
In \Cref{sec:discussion} we discuss the results in a broader context. 
Specifically, our results justify the usage of the \realRam.

%
%
\section{Algorithmic membership in \texorpdfstring{\ER}{ER}}
\label{sec:membership}

    The complexity class \ER is often called a 
    ``real analogue’’ of \NP, because it deals with real-valued
    variables instead of Boolean variables.
    This is because the real-valued ETR problem plays a role which is analogous to SAT when it comes to hardness and membership. 
    However, the most common way to think about \NP-membership is in terms of 
    certificates and verification algorithms.

    The seminal theorem of Cook and Levin shows the equivalence 
    of the two perspectives on~\NP-membership, see~\cite{levin1973universal,
    cook1971complexity}. We show a similar equivalence between
    \ETR-formulas and \emph{\realVerify{}s}.
    Intuitively, a \realVerify is an
    algorithm
    that runs on the \realRam
    that accepts as input both an \emph{integer}
    instance $I$ and a witness, and verifies that the witness describes a valid solution to the instance in polynomial-time.
    Note that as we define the \realRam as an extension of the \wordRam, it is also possible to guess and operate with integers.
    
    Formally we say a \emph{discrete decision problem} is any function $\QQ$ from arbitrary integer vectors to the booleans $\{\textsc{True},\allowbreak \textsc{False}\}$.
    An integer vector $I$ is called a \emph{yes-instance} of $\QQ$ if $\QQ(\Inst) = \textsc{True}$ and a \emph{no-instance} of $\QQ$ if $\QQ(\Inst) = \textsc{False}$.  Let $\circ$ denote the concatenation operator.
    A \emph{real verification algorithm} for $\QQ$ is
    formally defined as
    a \realRam algorithm $A$ that satisfies the following conditions, for some constant $c\ge 1$: (1) $A$ halts after at most $N^c$ time steps given any input of total length~$N$ where each value uses word size $\lceil c\log_2 N \rceil$. (2)
For every yes-instance $I\in\Z^n$, there is a real vector $\Rcert$ and an integer vector $\Zcert$, each of length at most $n^c$, such that $A$ \textsf{accept}s input $(\Rcert, \Inst\circ \Zcert)$ and (3) for every no-instance $I$, for every real vector $x$ and every integer vector $\Zcert$, $A$ \textsf{reject}s input $(\Rcert, \Inst\circ \Zcert)$.
A \emph{certificate} (or \emph{witness}) for yes-instance $\Inst$ is any vector pair $(\Rcert, \Zcert)$ such that $A$ accepts $(\Rcert, \Inst\circ\Zcert)$. We show the following theorem:
    \begin{restatable}{theorem}{CookLevin}
    \label{thm:Cook-Levin}
        For any discrete decision problem \QQ, 
        there is a \realVerify for \QQ 
        $\Leftrightarrow$ 
        $\QQ\in\ER$.
    \end{restatable}

    Our proof 
    in \Cref{sec:Cook-Levin}
    follows classical simulation arguments reducing
    nondeterministic polynomial-time (integer) random access machines
    to polynomial-size circuits or Boolean formulas, either
    directly~\cite{r-rfcs-91} or via nondeterministic polynomial-time
    Turing machines \cite{cr-tbram-73,e-mms-90,c-spfrnc-87,cook1971complexity}.
    We wish to state that the analysis in \Cref{sec:Cook-Levin} is broad enough to support the square root operator in the \realRam.
    
    The complexity class
    \ER\ is known to be equivalent to the discrete portion of the
    Blum-Shub-Smale complexity class $NP_{\mathbb{R}}^0$---real sequences that
    can be accepted in polynomial-time by a non-deterministic BSS machine
    \emph{without constants}, and the equivalence of BSS machines
    without constants and ETR formulas is already well-known
    \cite{blum1989theory,bcss-crc-98}.
    However as we explained in the introduction, the BSS-machine does not directly support the \emph{integer}
    computations necessary for common standard programming paradigms such as
    indirect memory access and multidimensional arrays.  The \realRam model
    originally proposed by Shamos \cite{s-cg-79,ps-cgi-85} \emph{does} support
    indirect memory access through integer addresses; however, Shamos did not
    offer a precise definition of his model, and we are not aware of any
    published definition precise enough to support a simulation result like
    \Cref{thm:Cook-Levin}.  We rectify this gap with our precise definition of the \realRam in the previous section together with a table of operations presented in \Cref{sec:DefRealRam}. Our proposal generalizes
    both the \wordRam and BSS models, we believe it formalizes the
    intuitive model implicitly assumed by computational geometers.
    
    \Cref{thm:Cook-Levin} not only strengthens the intuitive analogy
    between~\NP and~\ER, but also enables much simpler proofs of \ER-membership
    in terms of standard geometric algorithms.
    Our motivation for developing \Cref{thm:Cook-Levin} was 
    Erickson’s \emph{optimal curve straightening} problem~\cite{erickson2019optimal}:
    Given a closed curve~$\gamma$ in the plane (Given as a plane $4$-regular graph) and an integer $k$, is any
    $k$-vertex polygon topologically equivalent to~$\gamma$?
    (See \Cref{fig:CurveStraigthening}.)  The \ER-hardness
    of this problem follows from an easy reduction from stretchability of pseudolines,
    but reducing it directly to ETR proved much more complex; in light of
    \Cref{thm:Cook-Levin}, membership in \ER\ follows almost immediately
    from the standard Bentley-Ottman sweep-line algorithm \cite{bo-arcgi-79}.
    The theorem also applies to geometric packing problems~\cite{etrPacking}, where the input is a set of geometric objects and a container and the output is a pairwise disjoint placement of the objects in the container. Its \realVerify is straightforward.
    To further illustrate the power of our technique, we also consider a new
    topological problem in \Cref{sec:Cook-Levin}, which we call
    \emph{optimal unknotted extension}: Given a simple polygonal \emph{path} $P$ in $\R^3$
    and an integer $k$, can we extend $P$ to an \emph{unknotted} closed polygon
    with at most $k$ additional vertices?  
    Note that the path~$P$ only becomes a closed curve after we added the additional vertices.
    In light of \Cref{thm:Cook-Levin},
    the proof that this problem is in \ER\ is straightforward: To verify a positive
    instance, guess the $k$ new vertices and verify that the resulting polygon is unknotted
    using existing NP algorithms~\cite{hlp-ccklp-99,l-pubrm-15}.
    
    \begin{figure}
    \centering
    \includegraphics{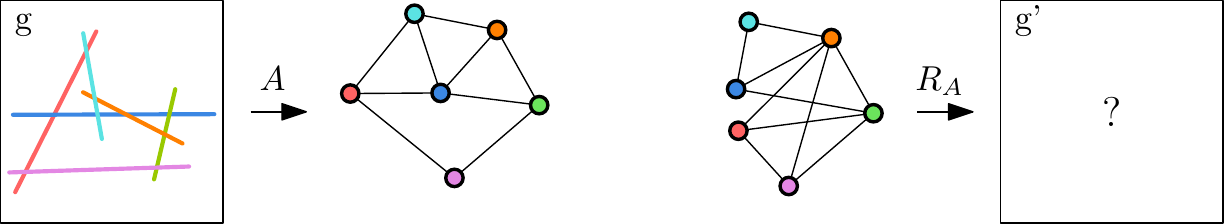}
    \caption{Left: given a set of segments $g$, algorithm $A$ constructs 
    segment intersection graph $c = A(g)$. \
    Right: given a graph $c'$, algorithm $R_A$ searches for a set of segments $g'$ such that $A(g') = c'$.}
    \label{fig:segments}
\end{figure}

    \begin{restatable}{corollary}{ApplicationLevin}
    \label{cor:ApplicationLevin}
        The following discrete decision problems are in \ER: The art gallery problem~\cite{ARTETR}, the optimal curve straightening problem~\cite{erickson2019optimal}, geometric packing, and the optimal unknotted extension problem.
    \end{restatable}
    
    Let us remark that researchers are sometimes using the sine or cosine functions as elementary operations of the \realRam. 
    So one may wonder whether, it is possible
    to extent the definition of the \realRam to include those functions as well. 
    We want to point out that the \realRam would then easily be powerful enough to decide undecidable problems, by Richardson's Theorem~\cite{richardson1969some}.
    Therefore, we refrain from using those functions for the \realRam.

\section{Smoothed Analysis of Recognition Problems}
\label{sec:Smoothed-Recognition}

    In computational geometry, we study many different
    geometric objects like point sets, polytopes,
    disks, balls, line-ar\-rangements, segments, and 
    rays. Many algorithms \emph{only} use combinatorial
    properties of the underlying geometry. 
    A recent impressive example is the EPTAS for 
    the clique problem on
    disk intersection graphs by Bonamy \etal~\cite{EptasClique}. In the paper they 
    first derive a set of properties for disk intersection
    graphs and then they use \emph{only} 
    those properties to find a new EPTAS.
    More classical results include methods for robust computations (e.g. computing the determinant for an orientation test, comparing a quadtratic polynomial for an in-circle test). Refer to a recent set of examples by Kettner~\etal~\cite{kettner2008classroom} or results on Exact Geometric Computation (EGC) such as the results by Yap~\cite{yap1997towards}.
    
    The above motivates that sometimes it is nice to reason about geometric problems based on only the combinatorial properties of the geometric configuration. 
    A related question, is whether an arbitrary combinatorial set, has an accompanying geometrical configuration. This question is the crux of a geometric recognition problem.

     Formally (\Cref{fig:segments}), we say  $A$ is a \emph{polynomial-time} \emph{\gvAlgo}
    if it takes some real-valued geometric input $g \in [0,1]^n$ and outputs a combinatorial object $A(g) \in \Z^m$ in time polynomial in $n$ (note that this implies that $m$ is polynomial in $n$). 
 
    Note that there is no need to give a witness for a negative answer. Although this maybe at first confusing, recall that \NP algorithms are only required to give a witness for positive answers.
    We define a \emph{recognition problem} as a discrete decision problem~$R_A$ that
    takes a combinatorial object~$c$ as input,
    and returns $\textsc{True}$ if there exists a vector $g \in [0,1]^n$ (which we shall call a \emph{witness}) for which $A(g) = c$.  For notational convenience, we denote for a given $c$ by $R_A(c)$ \emph{an arbitrary} witness of $c$ (as in the \SmoothAna that is to come, we consider the worst choice over all witnesses for $c$).

    \paragraph{\TotalPoly.}
    When we consider our applications of \SmoothAna to recognition problems, we can observe that regardless of the input or output, the total number of polynomials that may ever be evaluated by an algorithm that constructs the output is polynomially bounded in the number of (real) input variables. 
    To illustrate this, consider order type realizability. Given a set of $n$ points,
    an algorithm may need to check the order type of some triple of points. This may be done through verifying the sign of the determinant associated with that triple of points. 
    Since there are at most $O(n^3)$ different triples of points, there are at most $O(n^3)$ different polynomials that may be evaluated by an algorithm that identifies the order type of a point set. 
    A similar example exists for computing the intersection graph for objects in the plane. 
    For instance for unit disk intersection graphs, where it is sufficient to check for intersections between any pair of disks.
    Since there are at most $O(n^2)$ such pairs, any algorithm that identifies the disk intersection graph can evaluate only polynomials from a set of at most $O(n^2)$ polynomials. 
    For a given algorithm where the input consists of $n$ reals, we denote by $C(n)$ the total number of possible polynomials that may be checked by the algorithm.
    We say that an algorithm has a \emph{\fewPolynomials} if the \totalPoly $C(n) \leq n^{O(1)}$ (note that if some algorithm has running time $T(n)$ then $C(n) \leq 2^{T(n)}$.
    This upper bound comes from the fact that the algorithm branches at most $T(n)$ times. )

    In a previous version \cite{FOCSRobustComputation}
    of this article, we erroneously claimed that $C(n)$ is bounded from above  by~$T(n)$.
    To see an illustrative counter-example, consider the \textrm{even-interval} problem.
    In the \textrm{even-interval} problem we are given a real number $x \in [0,1]$ 
    and $k\in \N$ given in unary.
    We ask if $x$ is contained in one of the intervals 
    $I_i = [\frac{2i}{k}, \frac{2i+1}{k}]$, for some~$i$.
    For this problem it makes sense to measure the length of the input as
    the number of bits of~$k$.
    This problem can easily be solved with binary search in $O(\log k)$ steps where at every step we check one polynomial that verifies if $x$ lies in some interval $I_i$. 
    Yet, depending on the value of $x$, the set of intervals $I_i$ that we need to check wildly varies and is roughly $O(k)$. Hence $C(k) = O(k)$ and thus not bounded from above by the running time $O(\log k)$.
    Note that all of those polynomials have \AlgDeg and \AlgDim equal to one.
    Specifically, the example shows that \Cref{thm:Smooth-Real-RAM} is not true, without the assumption that there are \fewPolynomials.
    
    \paragraph{Defining \SmoothAna on recognition problems.}
\begin{figure}
    \centering
    \includegraphics{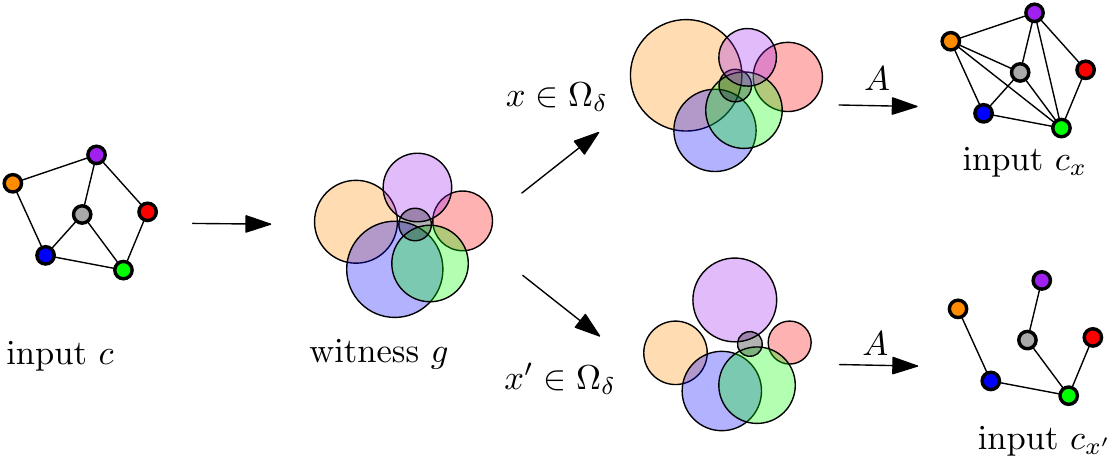}
     \caption{We define a perturbation on a combinatorial structure $c$ through a witness $g$.  A distribution of witnesses defines a distribution of discrete structures. }
    \label{fig:recognitionperturbation}
\end{figure}
    Traditionally in \SmoothAna one perturbs the input of an algorithm and measures the expected cost of executing the algorithm with the new input.  The real-valued geometric input $g$ offers a straightforward way to perturb it by adding to each $g_i \in [0,1]^n$ a random offset $x_i \in [\frac{-\delta}{2}, \frac{\delta}{2}]$ chosen uniformly at random. 
   It is not as easy to define a perturbation on the combinatorial (discrete) input that a recognition problem requires. 
    This is why we define the perturbation in terms of the geometric witness:
    given any input and witness $(c,g)$, we slightly perturb the witness $g$ to a new geometric object $g_x = g + x$ and
    reconstruct the corresponding combinatorial input $c_x$ (\Cref{fig:recognitionperturbation}). 
    A probability distribution over possible output $g_x$, combined with a recognition algorithm $A$, implies a probability distribution over all input $c_x$.

    We want for an instance $c$ of $R_A$, to find a witness $g$ of bounded \BitLength.  Note that the witness $g \in [0,1]^n \times \Z^0$. The witness $g$ is thus a valid input for a \realRam program with input size $n$ and $m= 0$. Our definition is formulated in terms of the \InputBitComp of $A$:
\[ 
\BIT_{\delta}(g, A) := \underset{x \in \Omega_\delta}{\E} 
\left[ \BIT_{IN}(g + x, A) \right], 
\]

\[
\BIT_{\delta}(c, R_A) := \underset{x\in\Omega_\delta}{\E} 
\left[ \BIT_{IN}(g_x = R_A(c) + x, A) \right]. 
\]
Note that the bit-complexity of $c$ is zero in case that 
there is no witness. Similarly to an \NP-algorithm running in constant time for a no-instance.

Next, we define the smoothed precision. We denote by $\Lambda_m$ the set of all combinatorial instances $c$ of size $m$ (that is, all combinatorial instances $c$ that can be described by a vector $b \in \Z^m$):
\[ \hspace{-1.8cm} \BIT_{\textrm{smooth}}(\delta, n, A) = \max_{g \in [0,1]^n} \; \BIT_{\delta}(g, A), 
\quad
\BIT_{\textrm{smooth}}(\delta, m, R_A) =  \max_{c \in \Lambda_m} \ \BIT_{\delta}(c, R_A).
\]
Now, we are ready to state our main theorem about smoothed analysis of recognition verification algorithms, whose proof can be found in \Cref{sec:Smoothing}. 

\begin{restatable}{theorem}{RealSmooth}
    \label{thm:Smooth-Real-RAM}
    {Let $A$ be
    a polynomial-time \gvAlgo with \AlgDeg $\Delta$, \AlgDim $d$ and few polynomials.}
    Under perturbations of $g \in [0,1]^n$ by $x \in \left[\frac{\delta}{2},\frac{\delta}{2} \right]^n$ chosen uniformly at random, the algorithm $A$ has a smooth
    \InputBitComp 
    of at most:
    \[ \BIT_{\textrm{smooth}}(\delta, n, A) = O\left( d \log  \frac{d \Delta n}{\delta}\right). \]
\end{restatable}
%
%
We stated our results only in terms of expected value, when in fact our proof also readily implies the same statement with high probability. 
We do this in order to follow the tradition in \SmoothAna that focuses more on expectations.
Note that statements \textit{with high probability} and \textit{in expectation} can be linked with Chebyshev's~inequality. 
Statements \textit{with high probability} do usually imply statements about \textit{expectations}, 
by some standard arguments.
 
The core idea of the proof of \Cref{thm:Smooth-Real-RAM} (presented at the end of this section and illustrated by \Cref{fig:Snapping}) is to consider the \gvAlgo $A$ with perturbed 
input $g_x = g + x$, where $g$ is an arbitrary value in $[0,1]^n$ and $x$ is a small perturbation  chosen uniformly at random in $\left[\frac{-\delta}{2}, \frac{\delta}{2}\right]^n$.
We model the perturbed input $g_x$ 
as a high-dimensional point which we snap to a fine grid to 
obtain $g'$ (input which can be described using bounded precision). 
We then show that for any algorithm $A$ that meets our prerequisites, 
 $\BIT_{IN}(g_x, A)$ is low with high probability. 
For the snapping we consider a sufficiently small $\width$ and we snap the point $g_x$ to a point in $\width \Z^n$. The \VoronoiDiagram of the points in $\width \Z^n$ forms a fine grid in $[0,1]^n$. 
As we explained in the introduction, the content of a \realRam register for a specific comparison instruction is per assumption the quotient of two
polynomials whose variables depend on the input. The core 
argument is that if the point $g_x$ lies in a \voronoicell of a point with limited word size 
which does not intersect the variety of either of the two polynomials, 
then the comparison instruction will be computed correctly. 
We bound  from above  the proportion of
\voronoicell{}s that are intersected by the variety of 
a polynomial in \Cref{thm:HittingCubes} in \Cref{sec:HittingCubes}.  

The algebraic proof of \Cref{thm:HittingCubes} has been placed in \Cref{sec:HittingCubes} to not distract from the main story line. 
%
By our assumption $A$ has \fewPolynomials.
The union bound over the \totalPoly bounds  from above  the chance that the execution of $A$ differs for the input $g'$ and $g_x$. 
The union bound will show that with high probability, for any perturbed input $g_x$, $\BIT_{IN}(g_x, A)$ is low. 
The complete proof of \Cref{thm:Smooth-Real-RAM} can be found in \Cref{sec:Smoothing}.
\Cref{thm:Smooth-Real-RAM} implies a result of \SmoothAna of recognition problems:
    
    \begin{restatable}{theorem}{SmoothedRecognition}
    \label{thm:Smooth-Recognition}
        Let $R_A$ be a recognition problem with \gvAlgo $A$. 
        If $A$ is a polynomial-time algorithm with constant \AlgDeg, \AlgDim and \fewPolynomials then 
        \[\BIT_{\textrm{smooth}}(\delta, m, R_A) = O(\log(m/\delta)) .\]
    \end{restatable}

 \begin{proof}
        This can be shown simply by 
        using the definition and the result of \Cref{thm:Smooth-Real-RAM}.
     \begin{align*}
       \hspace{-2.5cm} \BIT_{\textrm{smooth}}(\delta, m, R_A) &= \hspace{0.25cm} \max_{c \in\Lambda_m}\; \hspace{0.25cm}
         \underset{x\in\Omega_\delta}{\E} 
\left[\BIT_{IN}(g_x, A) \mid g_x = R_A(c) + x \right] \\
         &= \max_{g : A(g) \in \Lambda_m}\;
         \underset{x\in\Omega_\delta}{\E}
\left[\BIT_{IN}(g_x, A) \mid g_x = g + x \right] 
\end{align*}
    By definition of a \gvAlgo, if $c = A(g)$ has size $m$ then $g$ has size $n = \Theta(m^{\emph{const}})$
    for some
    fixed constant $\emph{const} > 0$. 
    Thus \Cref{thm:Smooth-Real-RAM} implies:

    \begin{align*}
            \BIT_{\textrm{smooth}}(\delta, m, R_A) &\leq  \max_{g \in [0,1]^{n}}\; \BIT_\delta(g, A) \\
            &= \BIT_{\textrm{smooth}}(\delta, n , A)\\
            &=  O(\log ({n} / {\delta})) = O(\log ({m^{const}} / {\delta})) = O(\log({m} / {\delta})). \; \quad
    \end{align*}
    \end{proof}
    
    \begin{restatable}{corollary}{CorRecognitionProblems}
    \label{cor:ImportantRecognitionProblems}
        The following recognition problems under uniform perturbations of the witness of magnitude $\delta$ have a smoothed \BitComp of $O(\log n/\delta)$: recognition of realizable order types~\cite{van2019smoothed}, disk intersection graphs, segment intersection graphs, ray intersection graphs and the Steinitz Problem in fixed dimension.
    \end{restatable}

        \begin{figure}[htbp]
    \centering
    \includegraphics{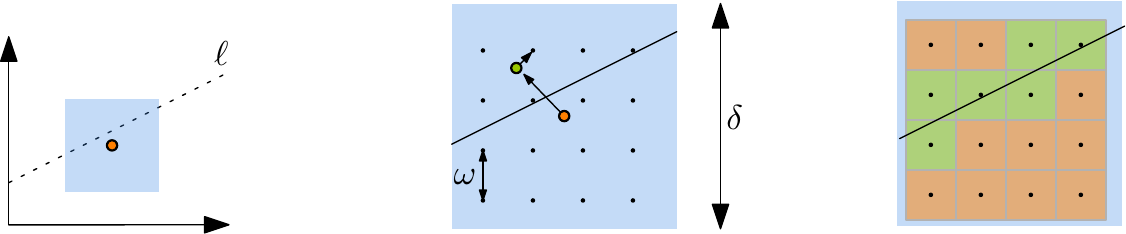}
    \caption{Given $g = (g_1,g_2) \in [0,1]^2$, we want to decide if the 
    point $g$ (in orange)
    lies above or below the line $\ell$ ($y = x/2 + 1$).
    If $g$ 
    is perturbed slightly to a point $g_x$, 
    low precision is usually sufficient. }
    \label{fig:Snapping}
    \vspace{-0.3cm}
\end{figure}

\paragraph{Limitations.}
    We want to point out that we cannot handle all recognition problems.
    First, not all recognition verification algorithms meet the
    conditions of \Cref{thm:Smooth-Real-RAM}.
    For example, consider the case that the problem deals
    with unbounded dimension. We still get some bounds on the \InputBitComp but they may be less desirable.
    A concrete example is the recognition of ball 
    intersection graphs,
    where the dimension of the ball is 
    part of the input. 
    Second, perturbing a witness may not
    be sensible. It does not mean that our theorems do not 
    apply in a mathematical sense, but rather that 
    in reality the result may be less desirable. For example, 
    this limitation applies to problems that rely on 
    degeneracies in one way or another.
    A concrete example is the point visibility 
    graph recognition problem. 
    Given a set of points $P$, we define a graph 
    by making two points  $p,q\in P$ adjacent if 
    line segment $pq$ contains no other point of $P$.
    This in turn defines a recognition problem where 
    we are given a visibility graph $G$ and we look for 
    a point set $P_G$ that realizes this visibility graph.
    If we perturb the real-valued witness $P_G$ then 
    with probability~$1$ there are no three collinear points. 
    Thus, the point visibility graph will always be  a clique.

\paragraph{Preliminaries for proving \Cref{thm:Smooth-Real-RAM}.}
Per definition, the \wordRam has a word size $\wordsize$ which allows us to express $2^\wordsize = \frac{1}{\omega}$ different values for each coordinate. 
We consider the \gvAlgo  $A$ with real-valued input $g \in [0,1]^n$. 
Since $A$ runs in polynomial-time it
can make at most a polynomial number of comparison operations (\Cref{tab:instructions}). 
At every such binary decision
the algorithm looks at a real- or integer-value register and verifies if the value at the 
register is $0$ or strictly greater than $0$.
For every real-valued register, per 
assumption the value at that register is the quotient of 
two $d$-variate polynomials $p_i$ and $q_i$ with maximum
degree $\Delta$ whose variables depend only on the values in the input register.
Let $z$ be the $d$-dimensional vector of input variables in $g$ that $p_i$ and $q_i$ depend on.

The evaluation of $p_i(z) / q_i(z)$ depends on the evaluation of the 
polynomials $p_i(z)$ and $q_i(z)$.
During \SmoothAna we perturb our input $g$ into new input $g_x = (g+x)$ with $x$ a value chosen uniformly at random in $[\frac{-\delta}{2},\frac{\delta}{2}]^m$.
Thereafter, we snap $g_x$ to $g'$
and we are interested in the 
probability that the execution of $A$ under both inputs ($g_x$ and $g'$ ) is the 
same, ergo the chance that for all comparison operations, 
the algorithms give the same answer.

    Fix a vector $z= (z_1,\ldots,z_d) \in \Z^d$ with integer coordinates. 
We denote by $C_{z}$ the unit
(hyper)cube, which has $z$ as its minimal corner: 
$C_z := [0,1]^d + z$. 
We denote by
$C(d, k) := \{ C_{z} \mid z \in \Z^d \cap [0, k)^d \}$ a $(k \times k \times \ldots \times k)$-grid of unit cubes that cover $[0, k]^d$.
Let $C = [0,k]^d$ be a cube partitioned by unit cubes $C(d,k)$. Every facet of $C$ intuitively is a grid of $(d-1)$-dimensional cubes $C(d-1, k)$. We argue about varieties that intersect cubes in $C(d,k)$ by induction on the dimension.

To that end, we define an equivalence relation on these sets of cubes. 
Let $C$ be a $d$-dimensional cube, partitioned 
by $d$-dimensional cubes of equal width. 
We say $C$ is \emph{equivalent} to $C(d,k)$, 
denoted by $C \cong C(d,k)$, if there exists 
an affine transformation $\tau$ of $C$ such that there 
is a one-to-one correspondence between cubes 
in $\tau(C)$ and $C(d,k)$ where corresponding cubes coincide. 
We give two examples of this equivalence relation that 
is often used in the remainder of the paper: 
(1) Consider any $d, k$ and any orthogonal $(d-1)$-dimensional, 
 hyperplane $H$ that intersects a $d$-dimensional 
cube $C(d,k)$ we have that $C(d, k) \cap H \cong C(d-1, k)$. 
(2) Consider the $d$-dimensional
grid $\Gamma_\width = \{C(y) : y \in \width\Z^d \cap [0,1]^d \} \cong  C(d,  1 /  \width)$, and a 
cube $C$ which has one corner on the origin and width 
$k \width$. The intersection $C \cap \Gamma_\omega$ 
is equivalent to $C(d, k)$.
Using these definitions and a well-known theorem by Milnor~\cite{basu2006algorithms} we obtain in \Cref{sec:HittingCubes} \Cref{thm:HittingCubes}.
See Theorem 7.23 in the printed second edition and Theorem 7.25 in the online version~\cite{basu2006algorithms}.

    \begin{restatable}{theorem}{HittingCubes}[Hitting Cubes, \Cref{sec:HittingCubes}]
\label{thm:HittingCubes}
    Let $p\neq 0$ be a $d$-variate polynomial 
    with maximum degree $\Delta$ and let $k \ge 2\Delta + 2$. 
    Then the variety $V(p)$ of $p$ intersects at most 
    $k^{d-1}\Delta 3^{d}(d+1)!$ unit cubes in $C(d, k)$.
    \end{restatable}

\section{Smoothed Analysis of Resource Augmentation}
\label{sub:resourceAugmentation}
    The predominant approach to find decent solutions for hard optimization
    problems is to compute an approximation. An alternative approach is resource augmentation, where you consider an optimal solution subject to slightly weaker problem constraints.
    This alternative approach has received considerably less attention in 
    theoretical computer science. We want to emphasize that resource augmentation 
    algorithms find a solution which does not compromise the 
    optimality, in the sense that it gives the optimal solution to the augmented problem. Using \SmoothAna, we argue that in practice some $\ER$-complete optimization problems have an optimal solution that can be represented with a logarithmic word size by applying \SmoothAna to the augmentation parameter. The proofs are deferred to \Cref{sec:ResourceAugmentation}.

For the art gallery problem, Dobbins, Holmsen and
Miltzow~\cite{ArxivSmoothedART} 
showed augmenting the polygon by so-called edge-inflations,
makes guarding the polygon easier.
This leads to small expected 
\InputBitComp under \SmoothAna. 
Their proof consists of a problem specific part
and a calculation of probabilities and 
expectations.
We generalize their idea to a widely applicable framework 
for \SmoothAna.

\paragraph{Definition.} 
Let us fix an algorithmic optimization problem $P$ with \realVerify $A$. 
For it to be
a \emph{resource-augmentation} problem,
we require three specific conditions:
First, each input consists of an $I \in [0,1]^n \times \Z^m$ together with 
an \emph{augmentation-parameter} $\alpha\in[0,1]$. 
Secondly, we assume that there is an implicitly defined
\emph{solution space} $\chi_I[\alpha] = \chi[\alpha]$ where each element in $\chi[\alpha]$ is considered a solution for the problem $P$ (with input $I$ and augmented by $\alpha$)  which does not have to be optimal. 
We require that for each $\alpha, \alpha'$ with $\alpha < \alpha'$ it holds that $\chi[\alpha] \subset \chi[\alpha']$.
For example, in the art gallery problem $I$ is the real-valued vertices of a simple polygon and an augmentation can be an inflation of the polygon by $\alpha$ (see \cite{ArxivSmoothedART}). The set $\chi_I[\alpha]$ is the set of all guard placements which guard the inflated polygon.
Thirdly, we assume that there exists an evaluation function
$f : \chi[1] \rightarrow \N$.
The aim, given $\alpha$, is to find a solution $x^* \in \chi[\alpha]$ for which $f(x^*)$ 
is the maximum or minimum denoted by $\Opt(\chi[\alpha])$.
Henceforth, for notational convenience, we assume that $P$ is a maximization problem.

\paragraph{Smoothed analysis of resource augmentation.}
For any $x$, we intuitively consider the minimal word size required for each coordinate in $x$ such that the verification algorithm $A$ of $P$ can verify if $x \in \chi_I[\alpha]$ on the \wordRam . 
We denote by $\BIT(\chi_I[\alpha]) = \min \{ \BIT_{IN}(x, A) \mid f(x) = \Opt(\chi_I[\alpha]), x \in \chi_I[\alpha]  \}$ the minimal precision 
needed to express an optimal solution in the 
solution space $\chi_I[\alpha]$.
For the \SmoothAna of resource augmentation, we choose $\alpha$ uniform at random in $[0, \delta]$ (since we assume that a negative augmentation is not well-defined). Just as in the previous section, we first define the expected cost of a given input:
$\BIT_\delta(I, P) =  \underset{\alpha \in\Omega_\delta}{\E} \left[ \BIT(\chi_I[\alpha])  \right]$ and the smoothed cost:
\[
\BIT_{\textrm{smooth}}(\delta, n, m, P) = \max_{I \in [0,1]^n \times \Z^m} \; \BIT_\delta(I, P).
\]
 
In this paper, we study resource augmentation problems with three additional but natural properties:

\begin{itemize}
    \item The \emph{\mono} property, which states that for all inputs $I$, for all $\alpha, \alpha' \in [0,1]$ if $\alpha \le \alpha'$
then $\chi_I[\alpha] \subseteq \chi_I[\alpha']$. Note that this property implies that in a more augmented version of the problem, the optimum is at least as good.
\item We define a breakpoint  as an $\alpha \in [0,1]$ such that $\forall \eps>0, \; \Opt(\chi_I[\alpha- \eps]) \neq \Opt(\chi_I[\alpha + \eps])$.
    Then the \emph{\breakproperty} property requires that there are at most~$n^{O(1)}$ breakpoints, for all inputs. 
   
\item  The \emph{\bitaugmented} property
    requires 
    that 
    for all 
    $x$ 
    optimal in $\chi_I[\alpha]$ and for all $\eps > 0$,  there is an
    $x'\in\chi_I[\alpha+\eps]$ with \InputBitComp with respect to its \realVerify of
    $\leq c \log ( n / \eps )$, 
    for some $c\in \N$,
    and $f(x) \le f(x')$ (recall that we assumed that $P$ is a maximization problem. Else the statement is symmetric). Ergo, given some $x$ we can increase the augmentation parameter
    by \eps, to obtain an equally 
    good solution $x' \in \chi_I[\alpha +\eps]$.
    Furthermore, $x'$ has low \InputBitComp. Note that $x'$ is not necessarily
    optimal for $\chi_I[\alpha +\eps]$.
    \end{itemize}

We apply \SmoothAna to resource-augmentation 
problems with these three properties by choosing 
uniformly at random the augmentation
$\alpha \in [0,\delta]$. In \Cref{sec:ResourceAugmentation} we prove the following theorem:

\begin{restatable}{theorem}{THMaugmentation}
    \label{thm:augmentation}
    Let $P$ be a resource augmentation problem
    that is 
    \mono, \breakproperty and \bitaugmented and $m \le O(n^c)$ then $\BIT_{\textrm{smooth}}(\delta, n, m, P)$ is bounded  from above by  $O(\log(n / \delta))$.
\end{restatable}

    \paragraph{Implications of \Cref{thm:augmentation}.}
    To illustrate the applicability
    of our findings, we give the following two corollaries. 
    The first result was already
    shown in~\cite{ArxivSmoothedART}.
    The art gallery problem has been shown to 
    be \ER-complete \cite{ARTETR}, and currently \ER-completeness 
    of the packing problems is in preparation~\cite{etrPacking}.
    Our results imply that apart 
    from near-degenerate conditions the solutions to 
    these problems have logarithmic \InputBitComp.
    
    \begin{corollary}
    \label{cor:ApplyAugmentation}
         Under perturbations of the augmentation of magnitude $\delta$, the following problems have an optimal solution with an expected \InputBitComp of $O(\log(n / \delta))$:         (1) the art gallery problem under perturbation of edge inflation~\cite{ArxivSmoothedART}, (2) packing polygonal objects into a square
        container under perturbation of the container side-length.
    \end{corollary}

\section{Discussion}
\label{sec:discussion}
Here, we discuss the implications of our results to the real RAM and to the gap between \ER and \NP. 

\paragraph{The \realRam}
In \Cref{sec:Cook-Levin}, we define formally a model of the \realRam. 
This model corresponds to the model that is intuitively used by researchers in computational geometry for decades.
Yet, it is still conceivable that one can design a polynomial time algorithm for SAT on the \realRam, even if P~$\neq$~\NP.
Yet, this paper gives several arguments that support the intuition that those pitfalls will not happen.
Although, the \realRam was invented purely in order to avoid to deal with 
cumbersome precision issues, its role in \Cref{thm:Cook-Levin} is completely different.
In particular, suddenly all algorithms analyzed on the \realRam help us now to establish \ER-membership, exactly because they do allow real inputs. 
Furthermore, assume that we have an algorithm \textrm{SAT-SOLVE} that can determine in polynomial time on a \realRam if a given SAT formula is true,
then \Cref{thm:Cook-Levin} implies co-\NP $\subseteq$ \ER.
In addition, if there were an algorithm that would solve true quantified Boolean formulas (TQBF) in polynomial time on a \realRam then \Cref{thm:Cook-Levin} would even imply $\ER = \PSPACE$.
Still, it seems difficult to rule out a polynomial time algorithm for SAT on the \realRam, as we can also not rule out a SAT algorithm on the \wordRam.

\paragraph{Smoothing the gap between \texorpdfstring{\NP}{NP} and \texorpdfstring{\ER}{ER}.}
\label{sub:verificationsummary}
 In \Cref{sec:membership} we strengthened the intuitive analogy
    between~\NP and~\ER by showing that for both of them membership is equivalent to the existence of a polynomial-time verification algorithm 
    in their respective RAM. 
    It is well-known that $\NP \subseteq \ER$, but it is  unknown if the two complexity classes are the same. The gap between the two complexity classes could be formed by inputs for \ER-complete problems, for which the witness cannot (to the best of our knowledge) be represented using polynomial \InputBitComp.
    In \Cref{sec:Smoothed-Recognition} and \Cref{sec:ResourceAugmentation} 
    we show that for a wide class of \ER-complete problems that have such an ``exponential \InputBitComp phenomenon'' their witness can, under smoothed analysis, be presented using logarithmic \InputBitComp. Thus, the gap between 
\ER and \NP (if it exists) is 
formed by near-degenerate input. 
    \begin{figure}[htbp]
    \centering
    \includegraphics[page = 3]{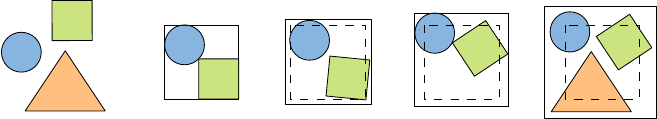}
    \caption{We augment the container from left to right. This
    extra space can lead to a better solution. If
    the optimal solution does not change, the extra
    space allows for a solution with low \InputBitComp.}
    \label{fig:ResourceAugmentation}
\end{figure}

\section{Algorithmic Membership in \texorpdfstring{\ER}{ER}}
\label{sec:Cook-Levin}


This section is devoted to proving the following theorem:

\CookLevin*

To this end, we finish the formalization of the \realRam that we started in the Introduction.

\subsection{What is the Real RAM?}
\label{sec:DefRealRam}

\begin{table}
 \centering\footnotesize\sffamily
 \def\arraystretch{1.3}
 \begin{tabular}{|c|c:c|}
 \hline
 	\textbf{Class} & \textbf{Word} & \textbf{Real} \\
 \hline
 	Constants
 			& $W[i]\gets j$			& $R[i] \gets 0$		\\
 			&						& $R[i] \gets 1$		\\
 \hdashline
 	Memory
 			& $W[i] \gets W[j]$		& $R[i] \gets R[j]$		\\
 			& $W[W[i]] \gets W[j]$	& $R[W[i]] \gets R[j]$	\\
 			& $W[i] \gets W[W[j]]$	& $R[i] \gets R[W[j]]$ 	\\
 \hdashline
 	Casting	& --- 			& $R[i] \gets j$	\\
 			& --- 			& $R[i] \gets W[j]$	\\
 \hdashline
 	Arithmetic and boolean
 			& $W[i] \gets W[j] \boxplus W[k]$	& $R[i] \gets R[j] \oplus R[k]$		\\
 \hdashline
 	Comparisons
 			& if $W[i] = W[j]$ goto $\ell$ & if $R[i] = 0$ goto $\ell$ \\
 			& if $W[i] < W[j]$ goto $\ell$ & if $R[j] > 0$ goto $\ell$ \\
 \hdashline
 	Control flow
 			& \multicolumn{2}{c|}{goto $\ell$} \\
 			& \multicolumn{2}{c|}{halt / accept / reject} \\
 \hline
 \end{tabular}
 \caption{Constant time RAM operations. The values $i,j,k$ are constant words used for indexing.}
 \label{tab:instructions}
 \end{table}

We explained in the introduction that the input to a \realRam algorithm consists of a pair of vectors $(a,b) \in \R^n \times \Z^m$, for some integers $n$ and $m$, which are suitably encoded into the corresponding memory arrays before the algorithm begins. To maintain uniformity, we require that neither the input sizes $n$ and $m$ nor the word size $w$ is known to any algorithm at “compile time”.  The output of a \realRam algorithm is the contents of memory when the algorithm executes the \textsf{halt} instruction.

Following Fredman and Willard \cite{fw-sitbf-93, fw-tams-94} and later users of the \wordRam, we assume that $w = \Omega(\log N)$, where $N = n+m$ is the total size of the problem instance at hand.  This so-called \emph{trans\-dichotomous} assumption implies direct constant-time access to the input data. %
\Cref{tab:instructions} summarizes the specific instructions our model supports.  All \emph{word} operations operate on words and produce words as output; all \emph{real} operations operate on real numbers and produce real numbers as output.  Each operation is parametrized by a small number of constant words $i$,~$j$, and~$k$.
The \emph{running time} of a \realRam algorithm is the number of instructions executed before its program halts; each instruction requires one time step by definition.

Our model supports the following specific word operations; all arithmetic operations interpret words as non-negative integers between $0$ and $2^w-1$.
\begin{itemize}
\item
addition: $x \gets (y+z) \bmod 2^w$
\item
subtraction: $x \gets (y-z) \bmod 2^w$
\item
lower multiplication: $x \gets (yz) \bmod 2^w$
\item
upper multiplication: $x \gets \floor{yz / 2^w}$
\item
rounded division: $x \gets \floor{y/z}$, where $z\ne 0$
\item
remainder: $x \gets y \bmod z$, where $z\ne 0$
\item
bitwise nand: $x \gets y \mathbin\uparrow z$ (that is, $x_i \gets y_i\mathbin\uparrow z_i$ for every bit-index $i$)
\end{itemize}
(Other bitwise boolean operations can of course be implemented by composing bitwise nands.)  
Similarly, our model supports the following \emph{exact} real operations.
\begin{itemize}
\item
addition: $x \gets y+z$
\item
subtraction: $x \gets x-y$
\item
multiplication: $x \gets y\cdot z$
\item
exact division: $x \gets y/z$, where $z\ne 0$
\item
(optional) exact square root: $x \gets +\sqrt{y}$, where $y\ge 0$
\end{itemize}

To avoid unreasonable computational power, our model does not allow casting real variables to integers (for example, using the floor function $\floor{\cdot}$), or testing whether a real register actually stores an integer value, or any other access to the binary representation of a real number.  However, we \emph{do} allow casting integer variables to reals. 

\subsection{Proofs of Algorithmic Membership}
This section is dedicated to prove \Cref{thm:Cook-Levin}.
As usual, we prove this theorem in two stages.  First, we describe a trivial real-verification algorithm for ETR.  Second, for any discrete decision problem $\QQ$ with a real-verification algorithm, we describe a polynomial-time algorithm \emph{on the \wordRam} that transforms every yes-instance of \QQ into a true \ETR formula and transforms every no-instance of \QQ into a false \ETR formula.  The first reduction implies that every problem in $\ER$ has a real-verification algorithm; the second implies that every real-verifiable discrete decision problem is in~$\ER$.

\begin{lemma}
ETR has a \realVerify.
\end{lemma}

\begin{proof}
Let $\Phi = \exists x_1, \dots, x_n \colon \phi(x_1, \dots, x_n)$ be an arbitrary formula in the existential theory of the reals.
The underlying predicate $\phi$ is a string over an alphabet of size $n+O(1)$ (the symbols $0, 1, +, \cdot, =, \leq , < , \land, \lor, \lnot$, and the variables $x_1, \dots, x_n$), so we can easily encode $\phi$ as an integer vector.  Our \realVerify takes the (encoded) predicate~$\phi$ and a real certificate vector $x\in \R^m$ as input, and evaluates the predicate $\phi(x)$ using (for example) a standard recursive-descent parser.  This algorithm clearly runs in time polynomial in the length of $\Phi$.
\end{proof}

\begin{lemma}
Every discrete decision problem with a \realVerify is in \ER.
\end{lemma}

\begin{proof}
Fix a \realVerify $A$ for some discrete decision problem $\QQ$.  We argue that for any integer vector $\Inst$, we can compute in polynomial-time \emph{on the \wordRam} a corresponding ETR formula $\Phi(I)$, such that $\Phi(I)$ is true if and only if $\Inst$ is a yes-instance of $\QQ$.  Mirroring textbook proofs of the Cook-Levin theorem, the formula $\Phi$ encodes the complete execution history of $A$ on input $(\Rcert, \textnormal{ } \Inst\circ\Zcert)$.  The certificate vectors $\Rcert$ and $\Zcert$ appear as existentially quantified variables of $\Phi$; the input integers $\Inst$ are hard-coded into the underlying proposition $\Phi$.

Now fix an instance $\Inst\in \Z^n$ of $\QQ$.  Let $N = n+2n^c$, let $w = \lceil c\log_2 N \rceil = c^2\log_2 n + O(1)$, 
and let $T = N^c = O(n^{c^2})$.  
Recall that the constant $c$ stems from the definition of a verifier.
Thus, $w$ is an upper bound on the word size and $T$ is an upper bound on the running time of~$A$ given input $(\Rcert, \Inst\circ\Zcert)$, for any certificates $\Rcert$ and $\Zcert$ of length at most $n^c$.  Our output formula $\Phi(I)$ includes the following \emph{register variables}, which encode the complete state of the machine at every time step~$t$ from $0$ to $T$:
\begin{itemize}
\item
For each address $i$, variable $\var{W(i,t)}$ stores the value of word register $W[i]$ at time $t$.
\item
For each address $i$, variable $\var{R(i,t)}$ stores the value of real register $R[i]$ at time~$t$.
\item
Finally, variable $\var{pc(t)}$ stores the value of the program counter at time~$t$.
\end{itemize}
Altogether $\Phi(I)$ has $(2\cdot 2^w + 1)T = O(n^{2c^2})$ register variables.  These are not the only variables in $\Phi(I)$; we will introduce additional variables as needed as we describe the formula below.

Throughout the following presentation, all indexed conjunctions ($\bigwedge$), disjunctions ($\bigvee$), and summations ($\sum$) are notational shorthand; each term in these expressions appears explicitly in the actual formula $\Phi(I)$.  For example, the indexed conjunction
\[
	\bigwedge_{b=1}^{w} \big(\var{2^b} = \var{2^{b-1}} + \var{2^{b-1}} \big)
\]
is shorthand for the following explicit sequence of conjunctions
\[
	\big(\var{2^1} = \var{2^0} + \var{2^0}\big) ~\land~
	\big(\var{2^2} = \var{2^2} + \var{2^1}\big) ~\land~ \cdots ~\land~
	\big(\var{2^w} = \var{2^{w-1}} + \var{2^{w-1}}\big).
\]

\smallskip
\paragraph{Integrality.}
To constrain certain real variables to have integer values, we introduce new global variables $\var{2^0}$, $\var{2^1}$, $\var{2^2}$, \dots, $\var{2^{w}}$ and equality constraints
\[
	\texttt{PowersOf2} :=
	(\var{2^0} = 1) \land \bigwedge_{b=1}^{w} (\var{2^b} = \var{2^{b-1}} + \var{2^{b-1}})
\]
The following ETR expression forces the real variable $X$ to be an integer between $0$ and $2^{w-1}$:
\[
	\texttt{IsWord}(X) ~:=~
	\exists x_0, x_1, \dots, x_{w-1} \colon 
			\left(X = \sum_{b=0}^{w-1} x_b \var{2^b}\right)\land
				\left(\bigwedge_{b=0}^{w-1} \big(x_b(x_b-1) = 0\big)\right)
\]
We emphasize that each invocation of \texttt{IsWord} requires $w$ new variables $x_b$; 
each $x_b$ stores the $b$th bit in the binary expansion of $X$.  
Although, this is a fairly lengthy way to check if a variable is an integer in a certain range, recall that ETR only has the constants $0,1$ in their alphabet. 
Our final formula $\Phi(I)$ includes the following conjunction, which forces the initial values of every word register variable to actually be a word:
\[
	\texttt{WordsAreWords} :=
		\bigwedge_{i=0}^{2^w-1} \texttt{IsWord}\left({\var{W(i,0)}}\strut\right)
\]
Altogether this subexpression involves $w2^w$ new one-bit variables and has length $O(w2^w)$.  Similarly, we can force variable $X$ to take on a fixed integer value $j$ by explicitly summing the appropriate powers of $2$:
\[
	\texttt{Equals}(X, j) ~:=~
			\left(X \!=\! \sum_{b \colon j_b = 1} \var{2^b}\right)
\]

\smallskip
\paragraph{Input and Output.}
We hardcode the fixed instance $\Inst$ into the formula with the conjunction
\[
	\texttt{FixInput} ~:=~ \bigwedge_{i=0}^{n-1} \texttt{Equals}\big({\var{W(i,0)}}, ~I[i] \big)
\]
(Here $I[i]$ denotes the $i$th coordinate of $I$.)  We leave the remaining initial word register variables $\var{W(i,0)}$ and all initial real register variables $\var{R(i,0)}$ unconstrained to allow for arbitrary certificates.

\smallskip
\paragraph{Execution.}
Finally, we add constraints that simulate the actual execution of $A$.  Let $L$ denote the number of instructions (“lines”) in program of $A$; recall that $L$ is a constant independent from $I$.  For each time step $t$ and each instruction index $\ell$, we define a constraint $\texttt{Update}(t, \ell)$ that forces the memory at time $t$ to reflect an execution of line $\ell$, given the contents of memory at time $t-1$.  Our formula~$\Phi(I)$ then includes the conjunction
\[
	\texttt{Execute} ~:=~
		\bigwedge_{t=1}^T \bigwedge_{\ell=1}^L
				\big(  (\var{pc(t)} = \ell) \Rightarrow \texttt{Update}(t, \ell) \big)
\]

The various expressions $\texttt{Update}(t, \ell)$ are nearly identical.  In particular, $\texttt{Update}(t, \ell)$ includes the constraints $\var{W(i,t)} = \var{W(i,t-1)}$ and $\var{R(j,t)} = \var{R(j,t-1)}$ for every word register $W[i]$ and real register $R[j]$ that are not changed by instruction $\ell$.  Similarly, unless instruction $\ell$ is a control flow instruction, $\texttt{Update}(t, \ell)$ includes the constraint
\[
	\texttt{Step}(t) ~:=~ \left(\var{pc(t)} = \var{pc(t-1)} + 1\right).
\]
\Cref{T:ram-encoding-easy,T:ram-encoding-indirect,T:ram-encoding-word-stuff} lists the important constraints in $\texttt{Update}(t, \ell)$ for three different classes of instructions.
\begin{itemize}
\item
Encoding constant assignment, direct memory access, real arithmetic (including square roots), and control flow instructions is straightforward; see \Cref{T:ram-encoding-easy}.  For an \textsf{accept} instruction, we set all future program counters to $0$, which effectively halts the simulation.  Similarly, we encode the \textsf{reject} instruction as the trivially false expression $(0=1)$.

\begin{table}[htb]
\centering\small\sffamily
\def\arraystretch{1.4}
\begin{tabular}{c:c}
	Instruction & Constraint 
\\
\hline
	$W[i] \gets j$			& $\texttt{Equals}(\var{W(i,t)}, j)$ \\
	$R[i] \gets 0$			& $\big(\var{R(i,t)} = 0\big)$\\
	$R[i] \gets 1$			& $\big(\var{R(i,t)} = 1\big)$\\
 	$R[i] \gets j$			& $\texttt{Equals}(\var{R(i,t)}, j)$\\
\hdashline
 	$W[i] \gets W[j]$		& $\big(\var{W(i,t)} = \var{W(j,t-1)}\big)$\\
	$R[i] \gets R[j]$		& $\big(\var{R(i,t)} = \var{R(j,t-1)}\big)$\\
 	$R[i] \gets W[j]$		& $\big(\var{R(i,t)} = \var{W(j,t-1)}\big)$\\
\hdashline
	$R[i] \gets R[j] + R[k]$		& 	$\var{R(i,t)} = \var{R(j,t-1)} + \var{R(k,t-1)}$
\\	$R[i] \gets R[j] - R[k]$		&	$\var{R(i,t)} = \var{R(j,t-1)} - \var{R(k,t-1)}$
\\	$R[i] \gets R[j] \cdot R[k]$	&	$\var{R(i,t)} = \var{R(j,t-1)} \cdot \var{R(k,t-1)}$
\\	$R[i] \gets R[j] / R[k]$		&	$\var{R(i,t)} \cdot \var{R(k,t-1)} = \var{R(j,t-1)}$
\\	$R[i] \gets \sqrt{R[j]}$		&	$\var{R(i,t)} \cdot \var{R(i,t)} = \var{R(j,t-1)}$
\\
\hdashline
 	if $W[i] = W[j]$ goto $\ell$	&
		$\IfThenElse{(\var{W(i,t-1)} = \var{W(j,t-1)})}
					{(\var{pc(t)} = \ell)}
					{\texttt{Step}(t)}$
\\
	if $W[i] < W[j]$ goto $\ell$  &
		$\IfThenElse{(\var{W(i,t-1)} < \var{W(j,t-1)})}
					{(\var{pc(t)} = \ell)}
					{\texttt{Step}(t)}$
\\
	if $R[i] = 0$ goto $\ell$ &
		$\IfThenElse{(\var{R(i,t-1)} = 0)}
					{(\var{pc(t)} = \ell)}
					{\texttt{Step}(t)}$
\\
	if $R[j] > 0$ goto $\ell$ &
		$\IfThenElse{(\var{R(i,t-1)} > 0)}
					{(\var{pc(t)} = \ell)}
					{\texttt{Step}(t)}$
\\
\hdashline
 	goto $\ell$ 				& $\var{pc(t)} = \ell$ \\
 	accept					& $\displaystyle\bigwedge_{i=t}^T (\var{pc(i)} = 0)$ \\
 	reject					& $0 = 1$ \\
\hline
\end{tabular}
\caption{$^{\strut}$ Encoding constant assignment, direct memory access, real arithmetic, and  control-flow instructions as formulae; “$\IfThenElse{A}{B}{C}$” is shorthand for {$(A\land B) \lor (\lnot A\land C)$}}
\label{T:ram-encoding-easy}
\end{table}

\item
Because there is no indirection mechanism in ETR itself, we are forced to encode indirect memory instructions using a brute-force enumeration of all $2^w$ possible addresses.  For example, our encoding of the instruction $W[W[i]] \gets W[j]$ actually encodes the brute-force linear scan “for all words $k$, if $W[i]=k$, then $W[k]\gets W[j]$”.  See \Cref{T:ram-encoding-indirect}.

\begin{table}[htb]
\centering\small\sffamily
\def\arraystretch{1.4}
\begin{tabular}{c:c}
	Instruction & Constraint \\
\hline
 	$W[W[i]] \gets W[j]$
							& $\displaystyle
								\bigvee_{k=0}^{2^w-1^{\strut}}
									\Big(\big(\var{W(i,t-1)} = k\big) \land
										\big(\var{W(k,t)} = \var{W(j,t-1)}\big)\Big)$
\\[3ex]
 	$W[i] \gets W[W[j]]$	
							& $\displaystyle
								\bigvee_{k=0}^{2^w-1}
									\Big(\big(\var{W(j,t-1)} = k \big) \land
										\big(\var{W(i,t)} = \var{W(k,t-1)}\big)\Big)$
\\[3ex]
	$R[W[i]] \gets R[j]$	
							& $\displaystyle
								\bigvee_{k=0}^{2^w-1}
									\Big(\big(\var{W(i,t-1)} = k\big) \land
										\big(\var{R(k,t)} = \var{R(j,t-1)}\big)\Big)$
\\[3ex]
	$R[i] \gets R[W[j]]$ 	
							& $\displaystyle
								\bigvee_{k=0_{\strut}}^{2^w-1}
									\Big(\big(\var{W(j,t-1)} = k\big) \land
										\big(\var{R(i,t)} = \var{R(k,t-1)}\big)\Big)$
\\
\hline
\end{tabular}
\caption{Encoding indirect memory instructions as formulae$^{\strut}$}
\label{T:ram-encoding-indirect}
\end{table}

\item
Finally, \Cref{T:ram-encoding-word-stuff} shows our encodings of arithmetic and bitwise boolean operations on words.  For addition and subtraction, we store the result of the integer operation in a new variable~$z$, and then store $z\bmod 2^w$ using a single conditional.  For upper and lower multiplication, we define two new \emph{word} variables $u$ and $l$, declare that $u\cdot 2^w+l$ is the actual product, and then store either $u$ or $l$.  Similarly, to encode the division operations, we define two new word variables that store the quotient and the remainder.  Finally, we encode bitwise boolean operations as the conjunction of $w$ constraints on the one-bit variables defined by \texttt{IsWord}.
\end{itemize}

\begin{table}[htbp]
\centering\small\sffamily
\def\arraystretch{1.4}
\begin{tabular}{@{}c:c@{}}
	Instruction & Constraint \\
\hline
 	$W[i] \gets (W[j] + W[k]) \bmod 2^w$	&
			$\begin{gathered}
				\exists z\colon (z = \var{W(j,t-1)} + \var{W(k,t-1)}) ~\land^{\strut}
				\\ 
				\big(
						\IfThenElse	{(z < \var{2^w})}
									{(\var{W(i,t)} = z)}\\
									{(\var{W(i,t)} = z - \var{2^w})}
					\big)
			\end{gathered}$				
\\[6ex]
 	$W[i] \gets (W[j] - W[k]) \bmod 2^w$	&
			$\begin{gathered}
				\exists z\colon (z = \var{W(j,t-1)} - \var{W(k,t-1)}) ~\land
				\\ \big(
						\IfThenElse	{(z \ge 0)}
									{(\var{W(i,t)} = z)}\\
									{(\var{W(i,t)} = z + \var{2^w})}
					\big)
			\end{gathered}$				
\\[4ex]
\hdashline
 	$W[i] \gets (W[j] \cdot W[k]) \bmod 2^w$	&
			$\begin{gathered}
				\exists u, l \colon
				\textsf{IsWord}(u) ~\land~ \textsf{IsWord}(l)^{\strut}
				~\land~ (\var{W(i,t)} = l)
				~\land~ \\
				(u\cdot \var{2^w} + l = \var{W(j,t-1)} \cdot \var{W(k,t-1)}) \\
			\end{gathered}$				
\\[4ex]
 	$W[i] \gets \floor{W[j] \cdot W[k] / 2^w}$	&
			$\begin{gathered}
				\exists u, l \colon
				\textsf{IsWord}(u) ~\land~ \textsf{IsWord}(l)
				~\land~ (\var{W(i,t)} = u)
				~\land~ \\
				(u\cdot \var{2^w} + l = \var{W(j,t-1)} \cdot \var{W(k,t-1)}) \\
			\end{gathered}$				
\\[2ex]
\hdashline
 	$W[i] \gets W[j] \bmod W[k]$	&
			$\begin{gathered}
				\exists q, r \colon
				\textsf{IsWord}(q) ~\land~ \textsf{IsWord}(r)
				~\land~ (\var{W(i,t)} = r)
				~\land~^{\strut} \\
				(r +  \var{W(k,t-1)} \cdot q = \var{W(j,t-1)} )
				\\
				~\land~ (r < \var{W(k,t-1)})
			\end{gathered}$				
\\[6ex]
 	$W[i] \gets \floor{W[j] / W[k]}$	&
			$\begin{gathered}
				\exists q, r \colon
				\textsf{IsWord}(q) ~\land~ \textsf{IsWord}(r)
				~\land~ (\var{W(i,t)} = q)
				~\land~ \\
				(r + \var{W(k,t-1)} \cdot q = \var{W(j,t-1)} )
				\\
				~\land~ (r < \var{W(k,t-1)})
			\end{gathered}$				
\\[4ex]
\hdashline \\[-2ex] 
 	$W[i] \gets W[j] \mathbin\uparrow W[k]$	&
			$\begin{gathered}	
				\textsf{IsWord}(\var{W(i,t)}) \land
					\textsf{IsWord}(\var{W(j,t-1)}) 
					\land\\
					\textsf{IsWord}(\var{W(k,t-1)}) \land^{\strut}{}
				\\[-1ex]
				\displaystyle
						\bigwedge_{b=0}^{w-1} \big(
							\var{W(i,t)}_b = 1 - \var{W(j,t-1)}_b \cdot \var{W(k,t-1)}_b
						\big)
			\end{gathered}$
\\[6ex]
\hline
\end{tabular}
\caption{$^{\strut}$ Encoding word arithmetic and boolean instructions as formulae, where “$\IfThenElse{A}{B}{C}$” is shorthand for $(A\land B) \lor (\lnot A\land C)$.}
\label{T:ram-encoding-word-stuff}
\end{table}

\paragraph{Summary.}
Our final ETR formula $\Phi(I)$ has the form 
\[
	\exists \text{[variables]}  \colon
	\texttt{PowersOf2}
	\land
	\texttt{FixInput}
	\land
	\texttt{WordsAreWords}
	\land
	\texttt{Execute}
	\land
	(\var{pc(T)} = 0)
\]
Now suppose $I$ is a yes-instance of $\QQ$.  If we set the initial register variables to reflect the input $(\Rcert, I\circ \Zcert)$ for some certificate $(\Rcert, \Zcert)$, then \texttt{Execute}  forces the final program counter $\var{pc(T)}$ to $0$, at the time step when $A$ accepts $(\Rcert, I\circ \Zcert)$).  It follows that $\Phi(I)$ is true.

On the other hand, if $I$ is a no-instance of $\QQ$, then no matter how we instantiate the remaining initial register variables, the \texttt{Execute} subexpression will include the contradiction $(0=1)$ at the time step when $A$ rejects.  It follows that $\Phi(I)$ is false.

Altogether, $\Phi(I)$ has $O(2^w (T + w) + T L w)$ existentially quantified variables and total length $O(2^w TL) = O(n^{2c^2})$, which is polynomial in $n$.  
(Recall that~$L$ denotes the number of instruction lines in the program of~$A$.)
Said differently, the length of $\Phi(I)$ is at most a constant times the \emph{square} of the running time of  $A$ on input $(\Rcert, I\circ \Zcert)$, where $\Rcert$ and $\Zcert$ are certificate vectors of maximum possible length.

We can easily construct $\Phi(I)$ in polynomial-time \emph{on the word RAM} by brute force.  We emphasize that constructing $\Phi(I)$ requires no manipulation of real numbers, only of symbols that represent existentially quantified real variables.
\end{proof}

\subsection{Examples}

To illustrate the usefulness of \Cref{thm:Cook-Levin}, we give simple proofs that four example problems are in \ER.  
Recall the problems as stated in \Cref{cor:ApplicationLevin}.
\begin{enumerate}
    \item The art gallery problem
    \item Optimal straight curve straightening problem
    \item Geometric packing
    \item Optimal unknotted extension problem.
\end{enumerate}
For two of these problems, membership in \ER was already known \cite{ARTETR,erickson2019optimal}; however, our proofs are significantly shorter and follow from known standard algorithms.  
We introduce the fourth problem specifically to illustrate the mixture of real and discrete non-determinism permitted by our technique.

\begin{proof}[Proof of \Cref{cor:ApplicationLevin}.]
Recall that the input to the art gallery problem is a polygon $P$ with rational coordinates and an integer $k$; the problem asks whether there is a set $G$ of $k$ guard points in the interior of $P$ such that every point in $P$ is visible from at least one point in $G$.  To verify a yes-instance, it suffices to guess the locations of the guards (using $2k$ real registers), compute the visibility polygon of each guard in $O(n\log n)$ time~\cite{heffernan1995optimal}, compute the union of these $k$ visibility polygons in $O(n^2k^2)$ time, and finally verify that the union is equal to $P$.  We can safely assume $k<n$, since otherwise the polygon is trivially guardable, so the verification algorithm runs in polynomial-time.

The optimal curve-straightening problem was introduced by the first author \cite{erickson2019optimal}.  The input consists of an integer $k$ and a suitable abstract representation of a closed non-simple curve $\gamma$ in the plane with $n$ self-intersections; the problem asks whether there is a $k$-vertex polygon $P$ that is isotopic to $\gamma$, 
meaning that the image graphs of $P$ and $\gamma$ are isomorphic as plane graphs.  
To verify a yes-instance of this problem, it suffices to guess the vertices of the $k$-gon $P$ (using $2k$ real registers), compute the image graph of $P$ in $O((n+k)\log n)$ time using a standard sweep-line algorithm~\cite{bo-arcgi-79}, and then verify by brute force that $P$ and $\gamma$ have identical crossing patterns.  
Again, we can safely assume that $k = O(n)$, since otherwise the curve is trivially straightenable, so the verification algorithm runs in polynomial-time.

Recall that the input for geometric packing is a container, the corresponding pieces and a number~$k$ of how many pieces we wish to pack. 
The real witness is a description of the corresponding movement of each piece.
The verification algorithm places all pieces and checks that there are no intersections between pieces and all pieces are contained in the container.
This can be done with some standard sweep-line algorithm~\cite{bentley1979algorithms}.

Finally, the input to the optimal unknotted extension problem consists of a polygonal path~$P$ in $\R^3$ with integer vertex coordinates, along with an integer $k$; the problem asks whether $P$ can be extended to an \emph{unknotted} closed polygonal curve in $\R^3$ with at most $k$ additional vertices.  Like the two previous problems, this problem is trivial unless $k<n$.  To verify an yes-instance of this problem, it suffices to guess the coordinates of $k$ new vertices (using $3k$ real registers), and then check that the resulting closed polygonal curve is unknotted in \emph{nondeterministic} polynomial-time (using a polynomial number of additional word registers), either using the normal-surface algorithm of Hass \emph{et al.}~\cite{hlp-ccklp-99}, or by projecting to a two-dimensional knot diagram and guessing and executing an unknotting sequence of Reidemeister moves~\cite{l-pubrm-15}.
\end{proof}

\section{Polynomials Hitting Cubes}
\label{sec:HittingCubes}
%
In the first part of this section, 
we bound  from above the number of unit cubes that a $d$-variate 
polynomial $p$ 
of bounded degree $\Delta$ can intersect in~$C(d, k)$.
In the second part of this section we finish the proof of \Cref{thm:Smooth-Real-RAM} by applying a well-known lemma for the sum of probabilities. 

Following Cox \etal~\cite{cox2006using} we define a $d$-variate 
polynomial $p$ in $x_1, \ldots, x_d$ with real coefficients
as a finite linear combination of monomials with real 
coefficients. Let  
$p \in \R[x_1, \ldots, x_d]$ denote the set of such polynomials.

We denote by $V(p) := \{ x \in \R^d : p(x) = 0 \}$ 
the \emph{variety} of $p$. 
For any subset $S \subset R^d$, we say that $p$ intersects
 $S$ if $S \cap V(p) \neq \emptyset$. 
 Given a set of polynomials $p_1,\ldots,p_k$,
 we denote their variety as
 $V(p_1,\ldots,p_k) = \bigcap_{i=1,\ldots,k} V(p_i)$.
For any expression $f$ which defines a function,
 we also use the notation $V(f) = \{x: f(x) = 0\}$, although
 it is not necessarily a variety.
 We say $f$ \emph{intersects} a set $S$, if 
 $V(f) \cap S \neq \emptyset$.
We need the notion of the \emph{dimension}
of a variety. 
We assume that most readers have some intuitive understanding,
which is sufficient to follow the arguments.
It is out of scope to define this formally in this
paper, so we refer to the book by Basu, Pollack 
and Roy~\cite[Chapter 5]{basu2006algorithms}.
Specifically, \Cref{lem:HyperplaneRestriction}
has to be
taken for granted.
Given a polynomial $p \in \R[x_1,\ldots,x_d]$, the linear polynomial $\ell\in \R[x_1,\ldots,x_d]$ 
is a factor of $p$, if there exists some 
$q\in \R[x_1,\ldots,x_d]$ such that $\ell \cdot q = p$.

\HittingCubes*
 
Our proof gives a slightly stronger, but more complicated upper bound. 
The proof idea is to consider 
the intersection between a cube $C_z \in C(d, k)$ and the polynomial $p$. Then either a connected 
component of $V(p)$ 
is contained in $C_z$ or $V(p)$ must intersect one of 
the $(d-1)$-dimensional facets of $C_z$.
In order to estimate how often the first situation 
can occur, we use a 
famous theorem by Oleinik-Petrovski/Thom/Milnor, 
in a slightly weaker form. 
See Basu, Pollack and Roy~\cite[Chapter 7]{basu2006algorithms}
for historic remarks and related results.

\begin{theorem}[Milnor~\cite{basu2006algorithms}]
\label{thm:Milnor}
    Given a set of $d$-variate polynomials $q_0,\ldots,q_s$ 
    with maximal degree $\Delta$.
    Then the variety $V(q_0,\ldots,q_s)$ 
    has at most $(2\Delta)^{d}$ connected components.
\end{theorem}
We use Milnor's theorem later for more than one polynomial.

We also need the following folklore lemma.
For more background on polynomials, see the book from Cox, Little, O'Shea~\cite{coxIdeals}.
Specifically Hilbert's Nullstellensatz, which can be found
as Theorem~$4.77$ in the book by Basu, Pollack 
and Roy~\cite{basu2006algorithms} is important.
\begin{lemma}[folklore]
\label{lem:HyperplaneRestriction}
    Let $p\in \R [x_1,\ldots,x_d]$ 
    be a $d$-variate 
    polynomial and 
    $H = \{x\in \R^d : \ell(x) = 0\}$  a $(d-1)$-dimensional
    hyperplane.
    Then $V(p)\cap H$ is the variety of a $(d-1)$-variate
    polynomial or $\ell$ is a polynomial factor of $p$.
\end{lemma}
In our applications, $\ell$ will be of the form $x_i =a$,
for some constant $a$.

\begin{proof}[Proof of \Cref{thm:HittingCubes}]
Note first that if $p$ has a linear factor $\ell$,
we decompose $p$ into $p = q\cdot \ell$ and
apply the following for $q$ and $\ell$ separately.
This works as the maximum degree of $q$ drops by one
and the maximum degree of $\ell$ equals $1$.
Thus for the rest of the proof, we assume that
$p$ has no linear factors, which in particular,
makes it possible to apply \Cref{lem:HyperplaneRestriction}.

Let us define $f( d)$ as 
the maximum
 number of unit cubes of $C(d,k)$ that can be intersected by 
 a $d$-variate polynomial $p\neq 0$ with maximal degree $\Delta$. 
We will first show that 
\begin{equation}
    f( 1) \le \Delta.    
    \label{eqn:CubeStart}
\end{equation}

Then we will show in a similar manner 
for every  $d$, $k$ and $\Delta$ holds that 
\begin{equation}
    f( d) \leq 2  f( d-1) \cdot d(k+1) + (2\Delta)^{d}.    
    \label{eqn:CubeRecursion}
\end{equation}
Solving the recursion then gives the upper bound of the theorem as follows: 
first, we show by induction that \Cref{eqn:CubeStart} 
and \Cref{eqn:CubeRecursion} 
imply $f(d) \leq (k+1)^{d-1}\Delta 2^d(d+1)!$.
\Cref{eqn:CubeStart} establishes the induction basis.
Using $2\Delta\leq k$, the induction step goes as follows:
\begin{align*}
    f(d) &\leq 2 f( d-1) \cdot d(k+1) + (2\Delta)^{d}\\
     &\leq 2 f( d-1) \cdot d(k+1) + (2\Delta)k^{d-1}\\
     &\leq 2 (k+1)^{d-2}\Delta 2^{d-1}(d)! \cdot d(k+1) + (2\Delta)k^{d-1}\\
     &= (k+1)^{d-1}(2 \Delta)  2^{d-1}(d)! \cdot d + (2\Delta)k^{d-1}\\
     &= (k+1)^{d-1}(2 \Delta) ( 2^{d-1}(d)! \cdot d  + 1)\\
     &\leq (k+1)^{d-1}\Delta 2^{d}(d)! \cdot (d  + 1)\\
     &= (k+1)^{d-1}\Delta 2^{d}(d+1)!
\end{align*}
Now using $k\geq 2\Delta + 2 \geq 3$, we can deduce that 
\[  (k+1)^{d-1} = \frac{(k+1)^{d-1}}{k^{d-1}} \cdot k^{d-1}
      \leq (1.5)^{d} k^{d-1}\]
This implies that 
$f(d) \leq k^{d-1}\Delta 3^{d}(d+1)!$.
It remains to show the validity of \Cref{eqn:CubeStart,eqn:CubeRecursion}.

If $p$ is a univariate polynomial of degree $\Delta$ then 
its variety  $V(p)$ is a set of at most $\Delta$ points and 
therefore $p$ can intersect  at most $\Delta$ disjoint 
unit intervals and 
this implies \Cref{eqn:CubeStart}.

\begin{figure}[bthp]
    \centering
    \includegraphics[page = 2]{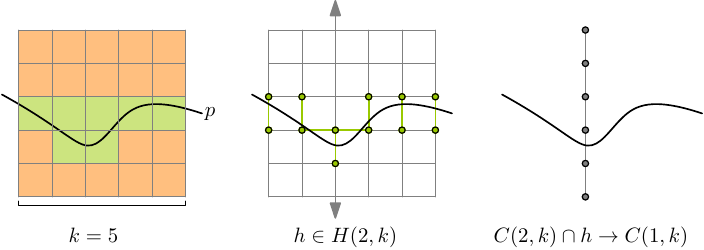}
    \caption{(Left) The polynomial $p$ whose variety intersects some cubes of $C(2, k)$. (Middle) The set $H(2,k)$ of axis-parallel lines and the $1$-dimensional facets of $C(2, k)$ that are intersected by $p$. (Right) The intersection of $C(2, k)$ 
    with a line $H$ gives $C(1, k)$. }
    \label{fig:IntersectGrid}
\end{figure}

To show the correctness of \Cref{eqn:CubeRecursion},
we refer to \Cref{fig:IntersectGrid}. 
Note that there are $d(k+1)$ axis-parallel 
$(d-1)$-dimensional hyperplanes 
with integer coordinates that intersect $C(d,k)$. We denote 
them by $H(d,k)$. 
Formally, we define the hyperplane 
$h(i,a) = \{\,x\in\R^{d} :  x_i = a \, \}$.
This leads to the definition:
$H(d,k) = \{ h(i,a): (i,a)\in [d]\times[k]\}$.
By the comment at the beginning of the proof,
we can assume that $p$ has no linear factors and
we apply \Cref{lem:HyperplaneRestriction}
on $p$ and each $h\in H(d,k)$.
For all cubes in $C(d,k)$, all facets of such a cube 
are contained inside
a $(d-1)$-dimensional hyperplane $h\in H(d,k)$.
By Milnor's theorem, there are at most $(2 \Delta)^d$ 
cubes in $C(d, k)$ which are intersected by $p$ but 
whose boundary is not intersected by $p$.
For any other cube in $C(d, k)$ that is intersected by the variety of $p$, 
the variety of $p$ must intersect a $(d-1)$-dimensional 
facet of that cube. 

Consider one of the $d(k+1)$ hyperplanes $h \in H(d,k)$. 
The set $I = h\cap C(d, k)$ is,
up to affine coordinate transformations, equivalent 
to $C(d-1, k)$. 
Furthermore, by \Cref{lem:HyperplaneRestriction},
$V(p)$ restricted to $h$ is the variety of a
$(d-1)$-dimensional polynom.
Thus by definition,
we know that $p$ intersects at most $f(d-1)$ 
cubes in $I$. 
Each of these $(d-1)$-dimensional cubes can 
coincide with a facet of at most two cubes in $C(d, k)$. 
It follows that $f( d)$ is bound  from above by 
$(2\Delta)^d + 2 \cdot f( d-1) \cdot k(d+1)$. 
This shows \Cref{eqn:CubeRecursion} and finishes the proof.
\end{proof}

\section{Smoothing the \realRam}
\label{sec:Smoothing}

This section is dedicated to prove the following theorem, which we first recall.
\RealSmooth*

For the \SmoothAna, we \emph{snap} our perturbed real-valued 
input $g_x = (g+x)$ onto a point with limited precision $g'$, the
closest point in $\width\Z^d \cap [0,1]^d$. Ergo, for all $y \in \width\Z^d$, all points in the \voronoicell of~$y$ 
are snapped to $y$. The \voronoicell{}s 
of all these cells (apart from those belonging to grid points on the boundary of $[0,1]^d$) are just $d$-dimensional cubes and we 
shall denote a cube centered at $y$ by $C(y)$.
We denote $\Gamma_\width = \{C(y) : y \in \width\Z^d \cap [0,1]^d \} \cong  C(d,  1 /  \width).$

The perturbed point $g_x$ must lie within some cube $C(y)$ in $\Gamma_\width$.
However the choice of original $g$ and $\delta$ limits the cubes in $\Gamma_\width$ where  $g_x$  can lie in. Specifically (\Cref{fig:subset}) the range $R$ of possible locations for $g_x$ is a cube of width 
$\delta$ centered around $g$. 
The boundary of this range cube $R$ likely does not align with the boundaries of grid cells in $\Gamma_w$. Therefore we make a distinction between two types of cells: the range cube $R$ must contain a cube of cells which is equivalent to $C(d, \lfloor \delta / \width \rfloor)$ and this sub-cube of $R$ we shall denote 
by $\Gamma_\width(g)$. All others cells which are intersected by $R$ but not contained in $\Gamma_\width(g)$ we call the \emph{perimeter cells}.
We can use this observation
together with \Cref{thm:HittingCubes} to estimate the probability
that $\sign(p(g_x)) \neq \sign(p(g'))$:

        \begin{figure}[htbp]
    \centering
    \includegraphics{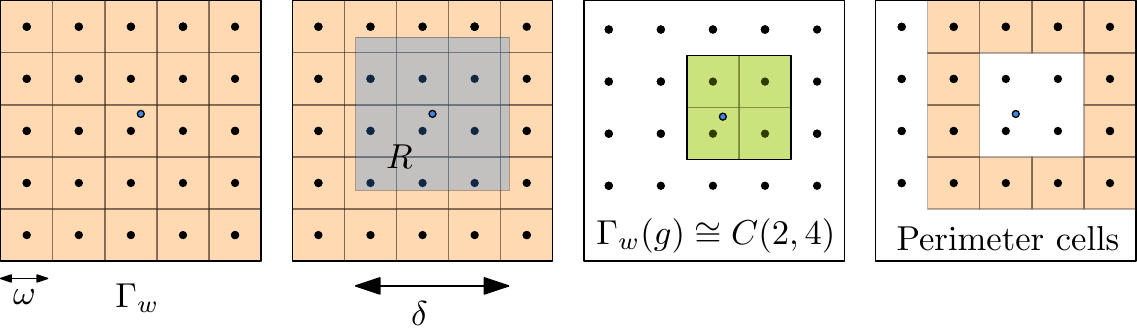}
    \caption{We depict the range $R$ of width $\delta$ around a point $g$. The block of green cells in $\Gamma_w(g)$ and the remaining cells that are intersected by $R$ are the perimeter cells.}
    \label{fig:subset}
    \vspace{-0.3cm}
\end{figure}

\begin{lemma}
    \label{lem:PolynomSnap}
    Let $p,q$ be two $d$-variate polynomials with
    maximum degree $\Delta$.
    Let $g \in \R^d$ be fixed
    and $x\in \Omega = [-\delta/2,\delta/2]^d$ chosen uniformly at random.
    Assume $g_x = g+x$ is snapped to a point $g'\in \width\Z^d \cap [0,1]^d$.
    Then $\sign\left(\frac{p(g_x)}{q(g_x)} \right) \neq \sign\left(\frac{p(g')}{q(g')}\right)$ with probability at most: 
    \[
    (4 \, \width \Delta 3^{d}(d+1)!) / \delta. \]
\end{lemma}
\begin{proof}
    It suffices to show that
    $\sign(p(g_x)) \neq \sign(p(g'))$ with probability at most:
    \[ \frac{2\, \width \Delta 3^{d}(d+1)!}{\delta}.\]
    The statement for $q$ is equivalent and the union bound on these separate probabilities then bounds  from above  the probability that their division does not have the same sign.
    For any polynomial $p$ and points $x, z \in \R^d$ with $p(x) <0$ and $p(z) >0$, there must be a point $y \in \line(x,z)$ for which $p(y) = 0$. 
    It follows that if a cube  $C \in \Gamma_\width$ does not intersect the variety of $p$, all points in $C$ either have a positive or negative evaluation under $p$. 
    Let $g'$ be the closest point to $g_x$ in $\omega \Z \cap [0,1]^d$ and let $C(g')$ be its \voronoicell.     If $C(g')$ does not intersect the variety of $p$ then $\sign(p(g_x)) = \sign(p(g'))$.
    Therefore we are interested in the probability that $g_x$ is contained in a Voronoi cell that is intersected by the variety of $p$. 
    
    As we discussed, the set of possible locations of $g_x$ is a cube which contains  $\Gamma_\width(g) \cong C(d, \lfloor \delta/\width \rfloor)$ together with a collection of perimeter cells. 
    
   We bound  from above the probability that $g_x$ lies within a cell that is 
    intersected by $p$
    in two steps: 
    (1) We  bound from above the number of perimeter cells and make the worst case assumption that they are all intersected by the variety of~$p$.
    (2)
    We bound  from above the number of cells of $\Gamma_\width(g)$ 
    that intersect the variety of $p$.
    These two numbers, divided by the total number of cells in 
    $\Gamma_\width(g)$ gives the upper bound we are looking for. 
    Note that the width of $\Gamma_\width(g)$ is equal to 
    $k = \lfloor \frac{\delta}{\width} \rfloor$ and that the perimeter of
    $\Gamma_\width(g)$ contains $2d \cdot k^{d-1}$ cells.  \Cref{thm:HittingCubes} gives an upper bound on the number of intersected cubes in $C(d, k)$~of:
    \[
    k^{d-1} \Delta 3^d (d+1)! \le 
    \left\lfloor \frac{\delta}{\width} \right\rfloor^{d-1} \Delta 3^d (d+1)!
    \]
    There are  $k = \lfloor \delta/\width \rfloor^d$ cubes 
    in $\Gamma_\width(a)$ and it follows that:
    \[
    \Pr\left(\sign(p(g_x)) \neq\sign(p(g'))\right) 
    \le \frac{(\lfloor \frac{\delta}{\width} \rfloor)^{d-1} \Delta 3^d (d+1)! + 2d(\lfloor \frac{\delta}{\width} \rfloor)^{d-1} }{\lfloor \delta/\width \rfloor^d} \le \frac{2 \width \Delta 3^d (d+1)!}{\delta}
    \]
    Using this, in conjunction with the union bound over the signs of $p(g_x)$ and $q(g_x)$ respectively, finishes the proof.
\end{proof}

\Cref{lem:PolynomSnap} is used to bound  from above the probability that snapping the perturbed input $g_x$ changes a single real-valued comparison in $A$.
We can use the union bound to bound  from above  the probability 
that for the whole algorithm, any real-valued comparison for $g_x$ 
and $g'$ is different.

Recall that we denote by $C(n)$ the total number of possible
polynomials with which a given algorithm~$A$ performs a comparison operation.
As remarked in \Cref{sec:Smoothed-Recognition}, it holds that $C(n)$ is bounded  from above  by $2^{T(n)}$,
where $T(n)$ denotes the running time.
The theorem statement of \Cref{thm:Smooth-Real-RAM} assumes $C(n)$ to be polynomial
in the input-size.

\begin{lemma}[Snapping]
    \label{lem:GridSnap}
    Let $g \in [0,1]^n$ be the input of an algorithm $A$ that 
    makes at most $C(n)$ potential \realRam comparison operations.
    Let $x$ be a perturbation chosen uniformly at random 
    in $[-\delta/2, \delta/2]^n$ and let $g_x = g + x$ 
    be a perturbed instance of the input. 
    For all $\eps \in [0,1]$, if $g_x$ is snapped to a grid of width  
    \[\width \leq  \frac{\eps \delta}{3^d (d+1)! \Delta 4 C(n) },\]
    then the algorithm $A$ makes for $g_x$ and $g'$ the same decision at each comparison instruction with probability at least $1-\eps$.
\end{lemma}

\begin{proof}
    By $E_i$ for $i \in [C(n)]$, we denote the event that 
    in the $i$'th polynomial on the inputs $g_x$ and $g'$ 
    get a different outcome. 
    \Cref{lem:PolynomSnap} bounds   from above the probability of $E_i$ occurring by:
    \[
    \Pr(E_i) \le \frac{4\width \Delta 3^d (d+1)!}{\delta}.
    \]
    The probability that $g_x$ and $g'$ are not equivalent is equal to the probability that for at least one event $E_i$, the event occurs. 
    In other words:
    \[
    \Pr(g_x \textnormal{ and } g' \textnormal{ not equivalent}) 
    = \Pr\left(\cup_{i=1}^{C(n)} E_i \right) 
    \le C(n) \cdot 4\width \Delta 3^d (d+1) / \delta.
    \]
     
    In the antecedent, $Pr(g_x \textnormal{ and } g' \textnormal{ not equivalent}) < \eps$ is implied by the inequality \[ C(n) \cdot \frac{4\width \Delta 3^d (d+1)}{\delta} <\eps,\] as stated above.
\end{proof}





    
    Now we are ready to bound  from above  the \emph{expected} \InputBitComp of $A$. 
    By definition it holds that
    \[ \E(\, \BIT_{IN}(g_x, A) \, ) = \sum_{k=1}^\infty k \Pr(\, \BIT_{IN}(g_x, A) = k \, ) \]
    
    Recall the following well-known lemma from Tonelli regarding the sum of expectations.
%
    Given a function 
    $f: \Omega \rightarrow \{1,2,\ldots  \}$ 
    and assume that $\Pr(f(x)>b) = 0$.
    Then it holds that
    \[\E[f] = \sum_{z=1}^{b} z\Pr(f(x) = z)  \ = \ \sum_{z=1}^{b} \Pr(f(x)\geq z).\]
 Using the lemma with $b\rightarrow \infty$, 
 the expected value of $\BIT_{IN}$ can be expressed as:
\[
\E(\, \BIT_{IN}(g_x, A) \, ) = \sum_{k=1}^\infty k \Pr(\, \BIT_{IN}(g_x, A) = k \, ) \  
= \sum_{k=1}^\infty \  \Pr(\, \BIT_{IN}(g_x, A) \geq k \, ) .
\]
We split the sum at a splitting point~$l$:
\[
\E(\, \BIT_{IN}(g_x, A) \, ) = \sum_{k=1}^{l} \Pr(\, \BIT_{IN}(g_x, A) \geq k \, ) \  +
\sum_{k=l+1}^{\infty} \Pr(\, \BIT_{IN}(g_x, A) \geq k \, ).
\]
Now we note that any probability is at most~$1$ therefore the left sum is at most $l$. 
Through applying \Cref{lem:GridSnap} we note that:
\[
\Pr(\BIT_{IN}(g_x, A) \geq k) = \Pr(\textnormal{GridWidth}(g_x) \geq 2^{-k}) \le 2^{-k} \left( \frac{3^d (d+1)! \Delta 4 C(n) }{\delta }\right),
\]
 
which in turn implies: 
\[
\E(\BIT_{IN}(g_x, A)) \le \sum_{k=1}^{l} 1 \  
+ 4 \left(\frac{3^d (d+1)! \Delta C(n) }{\delta}\right) 
\sum_{k=l+1}^{\infty} 2^{-k}.
\]
Observe that $\sum_{k=l+1}^{\infty} 2^{-k}  = 2^{-l}$. So if we choose 
$l = \lceil \log \frac{3^d  (d+1)! \Delta C(n) }{\delta} \rceil + 2$ we get:
\begin{align*}
\E(\BIT_{IN}(g_x, A)) &\le l +  4\left(\frac{3^d  (d+1)! \Delta C(n) }{\delta} \right) 2^{-l} \\
&\leq 
\left\lceil \log  \frac{3^d (d+1)!  \Delta C(n) }{\delta} \right\rceil
+ 2 +\\
& \vspace{1cm} +\left(\frac{3^d  (d+1)! \Delta C(n) }{\delta} \right) 
\left( \frac{\delta}{3^d  (d+1)! \Delta C(n) } \right) \\
&= 
\left\lceil \log \frac{3^d  (d+1)! \Delta C(n) }{\delta}  \right\rceil + 3 = 
 O\left( d \log \frac{d \Delta C(n) }{\delta}  \right) 
\end{align*}
Since this holds for arbitrary $g$ with $g_x = g + x$, this bounds  from above 
the \InputBitComp.
%
This finishes the proof of \Cref{thm:Smooth-Real-RAM}.

\section{Smoothed Analysis of resource augmentation}
\label{sec:ResourceAugmentation}
This section is devoted to proving \Cref{thm:augmentation} and its corollaries.
We briefly recall that for any $x$, we consider the minimal word size required for each coordinate in $x$ before a verification algorithm $A$ can verify if $x \in \chi_I[\alpha]$ on the \wordRam . We denote by $\BIT(\chi_I[\alpha]) = \min \{ \BIT_{IN}(x, A) \mid f(x) = \Opt(\chi_I[\alpha]), x \in \chi_I[\alpha]  \}$ the minimal precision 
needed to express an optimal solution in the 
solution space $\chi_I[\alpha]$.
We suppress $I$ in the notation when it is clear from context.

\THMaugmentation*

\begin{proof}
    Let the input $I$ be fixed and let $\alpha$ 
    be the augmentation parameter chosen uniformly 
    at random within the perturbation range $[0, \delta]$. 

    Consider any value $\eps > 0$.
   The variable $\alpha$ is chosen uniformly at random 
   in the interval $(0,\delta]$, and $P$ is \breakproperty 
   which implies that the probability that the interval 
   $[\alpha - \eps, \alpha]$ contains a breakpoint is 
   bounded  from above  by $ \eps \,  N / \delta $ for some choice of $N = n^{O(1)}$. 
    Assume that the interval $[\alpha - \eps, \alpha]$ 
    does not contain a breakpoint then by definition 
    $\Opt(\chi[\alpha - \eps]) = \Opt(\chi[\alpha])$. 
    Let $x$ be an optimal solution in $\chi[\alpha - \eps]$. 
    As the problem $P$ is
    \bitaugmented and $x \in \chi[\alpha - \eps]$, there must be an 
    $x' \in \chi[\alpha - \eps + \eps] = \chi[\alpha]$ with 
    an \InputBitComp of at most $c \log(n / \eps)$ and $x'$ 
    is also optimal for $\chi[\alpha]$.
    It follows that for a random $\alpha \in (0, \delta]$, 
    the probability that there is \emph{no} optimal 
    solution in $\chi[\alpha]$ with an \InputBitComp of 
    $c \log(n / \eps)$ is bounded  from above  by:
    \begin{equation}
    \label{eq:lowerAugment}
         Pr\left(\BIT(\chi[\alpha]) \ge c \log(n / \eps) \right) \le \frac{ \eps \cdot N}{ \delta}.
    \end{equation}
    Note that \Cref{eq:lowerAugment} holds for every \eps.
    We use this probability to obtain an upper bound 
    on the expected \InputBitComp. Using the lemma by Tonelli,
    we note that for all positive integers $l$ holds that:
    \[
\E(\BIT(\chi[\alpha])) = \sum_{k=1}^{l} \Pr(\BIT(\chi[\alpha] \geq k)) \  + \ 
\sum_{k=l+1}^{\infty} \Pr(\BIT(\chi[\alpha] \geq k)) 
\]

Any probability is at most $1$ which means that the 
left sum is bound  from above by $l$. We bound   from above 
$\Pr(\BIT(\chi[\alpha]) \ge k)$ by equating 
$   k = c\log(n/ \eps_k) \Leftrightarrow \eps_k = n/2^{(k/c)}$ and 
applying \Cref{eq:lowerAugment}:

\[
 \ \sum_{k=l+1}^{\infty} \Pr(\BIT(\chi[\alpha] \geq k)) 
=  \sum_{k=l+1}^{\infty} \Pr(\BIT(\chi[\alpha] \geq c\log(n/ \eps_k))) \le  \sum_{k=l+1}^{\infty}   \frac{\eps_k \cdot N}{\delta} \]

\[
 =   \frac{nN}{\delta}  2^{-(l+1)/c} 
\frac{1}{1-2^{-1/c}}  
 \le   2^{-(l+1)}  \frac{nN}{\delta(1-2^{-1/c})}     
\]


We choose $l = \lceil \log(n N / \delta) \rceil $ and note that:
    \[
\E(\BIT(\chi[\alpha])) \le \log(n N /  \delta) \ + \ 
2^{-\log(n N /  \delta)}  \ \frac{ n N}{\delta (1-2^{-1/c})}  + 
1 \le O(\log(n / \delta))
\]
 
This concludes the theorem.
\end{proof}

\paragraph{Applying \Cref{thm:augmentation}.}
Dobbins, Holmsen and Miltzow show 
that under smoothed analysis the \ER-complete art gallery
problem has a solution with logarithmic 
expected \InputBitComp~\cite{ArxivSmoothedART}. 
They show this under various models of perturbation 
with one being \emph{edge-inflation}. 
During \emph{edge-inflation}, for each edge $e$ 
of the polygon, they shift the edge
$e$ by $\alpha$ outwards in a direction orthogonal to $e$. 
Our theorem generalizes the \SmoothAna 
result from~\cite{ArxivSmoothedART}.

\begin{corollary}[\cite{ArxivSmoothedART}]
\label{cor:art}
    For the art gallery problem, with an $\alpha$-augmentation 
    which is an $\alpha$-edge-inflation, the expected \InputBitComp 
    with smoothed analysis over $\alpha$ is at most~$O(\log(n / \delta))$.
\end{corollary}

\begin{proof}
The \mono property and the \bitaugmented property are both shown in the short Lemmas 6 and 7 in \cite{ArxivSmoothedART}.
    The evaluation function $f$ in the art gallery problem 
    counts the number of guards of a solution. 
    The number of guards is a natural number and any 
    polygon can be guarded 
    using at most~$n/3$ guards~\cite{Fisk78a}. 
    This implies the \breakproperty property.
    Together with the \mono property this proves that 
    the number of breakpoints is bounded  from above  by $n/3$
    and the corollary follows.
\end{proof}

\begin{figure}[bthp]
    \centering
    \includegraphics{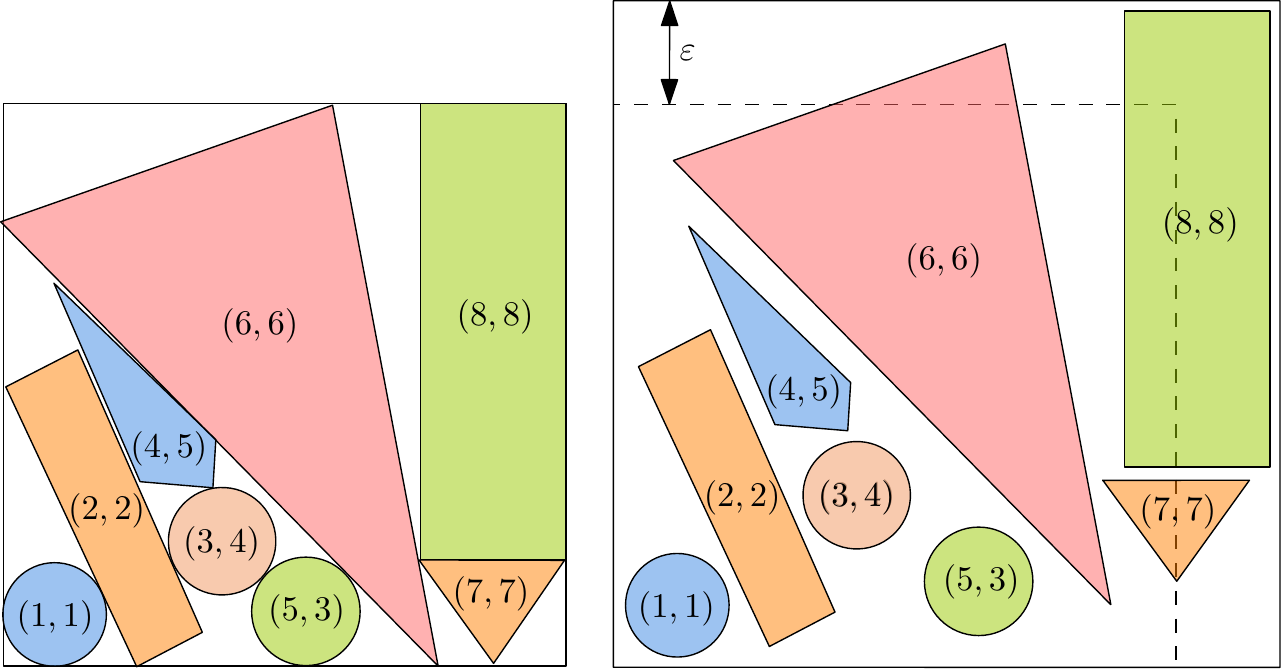}
    \caption{(Left) A set of 8 convex objects which are tightly packed in a square with diameter $(1+\alpha)$. Their enumerations specify a choice for $\pi_x$ and $\pi_y$. Note that this choice is not unique. (Right) The objects packed in a square of diameter $(1+\alpha + \eps)$ where an object with enumeration $(i,j)$ is translated by $(\frac{i\eps}{10}. \frac{j\eps}{10})$.   }
    \label{fig:shifting}
\end{figure}

We can also prove new results such as the following result about geometric packing:

\begin{corollary}
\label{cor:packing}
    For geometric packing of convex objects in a unit-size container, with an $\alpha$-augmentation over the container size, the expected \InputBitComp with smoothed analysis over $\alpha$ is $O(\log(n / \delta))$.
\end{corollary}

\begin{proof}
    In order to prove the corollary, we only have to prove that geometric packing with $\alpha$-augmentation over the container size is \mono, \breakproperty and \bitaugmented. For an input $I$ and unit container size the solution space $\chi_I[\alpha]$ is the set of all solutions which contain a set of geometric objects from $I$ packed within a container of size $(1+\alpha)$. 
    First, we show that this algorithmic problem has the \mono property.
    Specifically, we show that for all $\alpha, \alpha' \in [0,1]$ with $\alpha \le \alpha'$, $\chi_I[\alpha]$ contains subsets of $I$ which can be packed in a container of size $(1+\alpha)$ and therefore the elements can be packed in a container of size  $(1+\alpha')$.

Next, we show the problem is \breakproperty.
The evaluation function $f$ counts the number 
of geometric objects which are correctly packed 
in the solution. Thus the function $f$ can evaluate 
to at most $n$ distinct values and therefore the number 
of breakpoints is bounded from above  by~$n$.
    
    Lastly we show {\bitaugmented} property.
    We show that for all optimal packings $x$ in $\chi_I[\alpha]$ and for all $\eps > 0$,  there is a solution
    $x'\in\chi_I[\alpha+\eps]$ with \InputBitComp with respect to its \realVerify of
    $O( \log ( n / \eps ) )$
    and $f(x) \le f(x')$ (since packing is a maximization problem).
    In other words: given an optimal packing $x$, that packs $k$ objects in a container of size $(1 + \alpha)$ we must show that there is a packing $x'$ in a container of size $(1 + \alpha + \eps)$, that packs at least $k$ objects \emph{and} that the solution $x'$ can be specified using $O( \log ( n / \eps ))$ bits per coordinate.
    
    To this end, let $x \in \chi[\alpha]$ 
be an optimal packing of $k$ elements of the input in a 
container of size $(1+\alpha)$.
Fix a value $\eps>0$ and consider $x$ placed in a container of width $(1 + \alpha + \eps)$. Using the extra space and the fact that the objects are convex, we will translate the objects in $x$ such that any two objects in $x$ are at least $\frac{\eps}{n+2}$ apart from each other and the boundary by translating them in the two cardinal directions (refer to \Cref{fig:shifting}). 

Specifically, we assign a linear order $\pi_x$ on the $n$ elements such 
that: (1) if an object dominates another in the $x$ direction 
it is further in the ordering and (2) if you take an 
arbitrary horizontal line, the order of intersection with this 
line respects the order $\pi_x$. We define $\pi_y$ 
symmetrically 
for vertical lines. Since the input objects are convex, such a linear extension to $\pi_x$ must always exist (Figure~\ref{fig:pointplacement}). Let a convex object $O$ in 
$x \in \chi_I[\alpha]$ be the $i$'th element in $\pi_x$ and 
the $j$'th element in $\pi_y$ (we start counting from $1$). 
We translate $O$ by $\left( \frac{i\eps}{n+2}, \frac{j \eps}{n+2} \right)$. 
Observe that:
\begin{enumerate}
    \item all objects are contained in a container 
of diameter $(1 + \alpha + \eps)$ (since an element 
is translated by at most $\frac{n \eps}{n+2}$ 
in any cardinal direction) and
\item all objects 
are separated by at least $\frac{\eps}{n + 2}$ due to the translation. Indeed, since all objects were originally disjoint their distance was originally at least zero. Now they are separated by at least $\frac{\eps}{n + 2}$ in both the $x$ and $y$ direction.
\end{enumerate}
 This proves that we can place the packing $x$ in a container of size $(1 + \alpha + \eps)$, such that any two objects are separated by at least $\frac{\eps}{n+2}$.
 What remains to show, to prove that we can always specify this packing using logarithmic bits per coordinate.

Since any two objects are separated by at least $\frac{\eps}{n+2}$, we can freely translate every object by $O(\eps / n)$ and freely rotate every object by $O(\eps / n)$ degrees such that all objects are still mutually disjoint.
We claim that we can now move and rotate every object in $x$ by at most $O(\eps / n)$, such that all objects are still pairwise disjoint and such that their rotation and translation can be described with $O(\log n / \eps )$ bits. 
For translation, this is immediately true (simply snap one of the points of the convex object to a corresponding nearest gridpoint).
When describing a rotation however, it is not immediately obvious that one can describe a suitable rotation with $O(\log n / \eps )$ bits. As we are not aware of any publication that specifies how to describe rotations with limited \BitComp, we give an example of how to do this in Section~\ref{sub:rotations}. 

Thus, we showed that given any optimal packing $x \in \chi[\alpha]$ there exists a packing $x' \in \chi[\alpha + \eps]$, that packs at least as many objects and can be described with $O(\log n / \eps)$ bits per objects. Therefore, geometric packing is \bitaugmented. 
\end{proof}

\begin{figure}[h]
    \centering
    \includegraphics{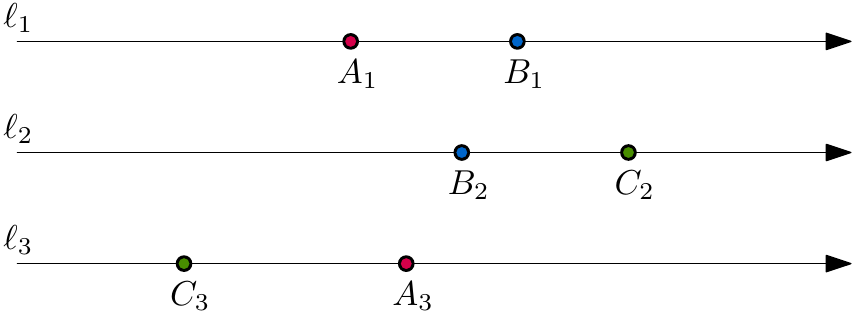}
    \caption{Consider for the sake of contradiction three planar pairwise disjoint convex objects $A, B, C$ with $\pi_x(A) < \pi_x(B)$, $\pi_x(B) < \pi_x(C)$ and $\pi_x(C) < \pi_x(A)$. Then, without loss of generality (due to symmetry), there exist three lines $\ell_1$, $\ell_2$, $\ell_3$ as shown in the figure, where $A$ intersects $\ell_1$ before $B$ in the intersection point $A_1$, $B$ intersects $\ell_2$ before $C$ in the intersection point $B_2$ and $C$ intersects $\ell_3$ before $A$ in the point $C_3$.
    Since $A$ is convex, there exists a segment connecting $A_1$ and $A_3$ contained in $A$. 
    This segment has to cross $\ell_2$, left of $B_2$ (else this segments intersects the segment between $B_1$ and $B_2$, and via therefore the convex object $B$).
    It follows that the segment connecting $A_1$ and $A_3$ intersects the segment connecting $C_2$ and $C_3$ which contradicts that these objects were pairwise disjoint.
    }
    \label{fig:pointplacement}
\end{figure}
 
Lastly to demonstrate the wide applicability of 
\Cref{thm:augmentation} we investigate a classical 
problem within computational geometry which is not an \ETR-hard 
problem: computing the minimum-link path. In this problem the 
input is a polygonal domain and two points contained in it and 
one needs to connect
the points by a polygonal path with minimum number of edges. 
Recently, Kostitsyna, L{\"o}ffler, Polishchuk and
Staals~\cite{kostitsyna2017complexity} showed that even 
if the polygonal domain $P$ is a simple polygon where the $n$~vertices
of the polygon each have coordinates with $\log n$ bits each, 
then still the minimal \InputBitComp needed to represent the 
minimum-link path is $O(k \log n)$ where $k$ is the length
of the path and they present 
a construction where $k = \Omega(n)$.

Just like the art gallery problem, the minimum link path problem has a simple polygon as input and we propose to augment the minimum link path problem in the same way as the art gallery problem was augment in~\cite{ArxivSmoothedART}: by edge-inflation. Two points in the minimum link path may be connected if and only if they are mutually visible. Hence, with an analysis identical to \Cref{cor:art} we can immediately state that:
\begin{corollary}
\label{cor:minlinkpath}
    For computing the minimum link path, with an $\alpha$-augmentation which is edge-inflation, the expected \InputBitComp with smoothed analysis over $\alpha$ is $O(\log(n / \delta))$.
\end{corollary}

\subsection{Describing a rotation using limited \InputBitComp}
\label{sub:rotations}

To finalise the argument for Corollary~\ref{cor:packing}, we claimed that for every object in a packing $x$, we could slightly rotate that object so that its rotation can be described with a limited number of bits. 
For completeness, we dedicate this section to providing a way to describe such a rotation. 

In computer applications (and specifically applications that contain geometry), it is not uncommon to want to rotate geometric objects.
However, describing a rotated geometric object is nontrivial on a \realRam and problematic on a \wordRam. More often than not, rotations are the result of user input and from an application perspective it would be nice if a user could specify an angle $\phi$ and obtain a geometric object that is rotated by $\phi$ degrees. However, this is infeasible in both the \wordRam \emph{and} \realRam.
Rotating a vector $v$ by $\phi$ degrees is equal to multiplying the vector $v$ with the matrix:
\[
 R 
=
\begin{pmatrix}
 \cos \phi & -\sin \phi \\
\sin \phi & \cos \phi
\end{pmatrix} 
\]

This rotation matrix $R$ is a matrix with two key
properties: (1) the transposed matrix $R^T$ is equal to 
the inverse of $R$,i.e., $R^T = R^{-1}$ and 
(2) the determinant of $R$ equals~$1$.
These two properties 
ensure that the rotated version of a vector still has the
same length and that adjacent edges that get rotated still
have the same angle between them. 
In order to compute $R$, given $\phi$, requires 
the computation of the $\cos$ and $\sin$ function, which
cannot be simulated by the \realRam, as explained
in \Cref{sec:DefRealRam}.
If we use the \wordRam, we can compute an approximation
of $\cos$ and $\sin$, but the resulting matrix, will
only in exceptional cases have properties~(1) and~(2).
However, note that we can specify any rotation $R$ by providing a unit
vector $V_R$ with:
\[
R \cdot  \begin{pmatrix}
 1\\
0
\end{pmatrix}  = v_R  \quad \Rightarrow \quad
    \begin{pmatrix}
 a & -b\\
b & a
\end{pmatrix} \cdot   \begin{pmatrix}
 1\\
0
\end{pmatrix}   = \begin{pmatrix}
 a\\
b
\end{pmatrix} = {v}_R. 
\]

Clearly the vector $v_R$ might have real coefficients and therefore a rotation that can be expressed in the \realRam may not always be used in the \wordRam. However, we show that for any $\wordsize$ and any desired rotation $v_R \in \mathbb{R}^2$ we can express a rotation matrix $R'$ which uses an \InputBitComp of $\wordsize$ bits such that: $|| R \cdot (1,0) - R'(1,0) ||< 2^{-w}$.
That is the goal of this subsection, we show that it is possible to obtain such a $R'$ with logarithmic \emph{\InputBitComp} such that the resulting rotation that we obtain is not far off.

Let $v_R = (a, b)$ be the desired rotation and let $v_R$ not be an orthonormal vector (if it is, then the desired rotation can be expressed using a constant-complexity rotation matrix). $V_R$ must lie in one of the four quantiles of the plane (We chose the quantiles to lie diagonally to the $x$ and $y$ axis, see \Cref{fig:rotations}). 
Given $v_R$, we consider the unique orthonormal vector that lies in the quantile opposite of $V_R$. Let without loss of generality $v_R$ be a vector lying in the left quantile, then the orthonormal vector that we consider is $(1, 0)$.
We now assume that $v_R$ lies above the $x$-axis, the argument for below the $x$-axis is symmetrical.

We round the coordinates $(a,b)$ to the closest point $(a', b')$ that lies above the $x$-axis, outside of the unit circle \emph{and} $a'$ and $b'$ can be described using \wordsize bits each.
Observe that since $a,b \in [0,1]$, it must be that $|| (a,b) - (a', b') || \le 2 \cdot 2^{-w}$.

    \begin{figure}[htbp]
        \centering
        \includegraphics[]{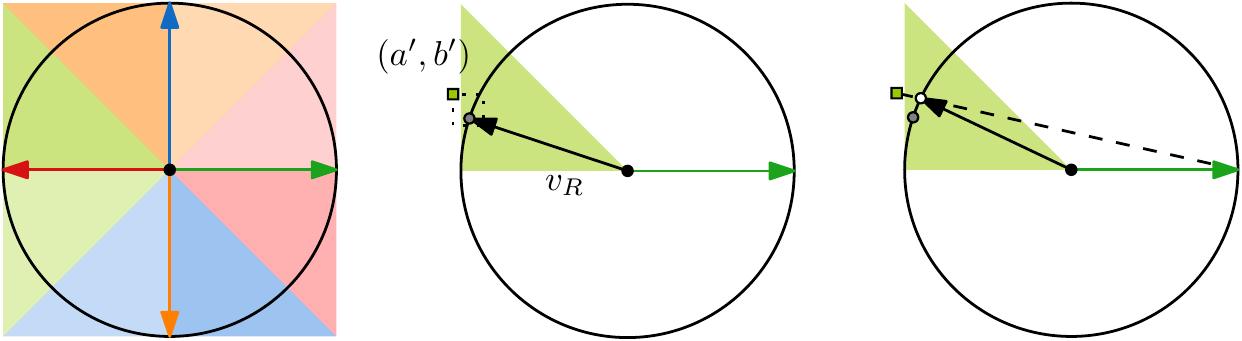}
        \caption{
        (Left) the four quantiles that we consider in orange, green, blue and red. If the desired vector $v_R$ lies in a quantile, we construct a rotation using the orthonormal vector in the opposite quantile. (Middle) We round $v_R = (a,b)$ to the closest top-right point that can be expressed with \wordsize bits. (Right) The points $(1,0)$ and $(a', b')$ and their point of intersection on the unit circle.   }
        \label{fig:rotations}
    \end{figure}

\begin{lemma}
\label{lem:rotation}
Consider the line through $(a', b')$ and $(1, 0)$ and its point of intersection with the unit circle $(x, y)$. It must be that:
\begin{enumerate}
    \item $|| (a, b) - (x, y) || \le 2 \cdot 2^{-w}$.
    \item The rotation matrix  $R' (1, 0) = (x,y)$ can be described using limited \InputBitComp.
\end{enumerate}
\end{lemma}

\begin{proof}
The proof is illustrated by \Cref{fig:rotations}.
Consider the square $C$ centered around $(a,b)$ with $(a',b')$ as its top left corner. Both of its right corners must be contained within the unit disk.
The line $\ell$ through $(1,0)$ and $(a',b')$ is equal to $y = \frac{b'(x-1)}{a' - 1}$. Per construction, the values $a'$ and $b'$ can be expressed using \wordsize bits, therefore the line can be described by $y = \alpha x - \beta$ with $\alpha$ and $\beta$ having $O(\wordsize)$ bits. Per construction, the slope $\alpha$ is between $0$ and $-1$ (since $(a', b')$ must lie within the same quantile as $(a,b)$ and, above the $x$-axis). It follows that the line $\ell$ intersects $C$ in its right facet and therefore the point of intersection $(x,y)$ must be contained within $C$ which proves the first item.

Item two: The unit circle is given by $y^2 + x^2 = 1$. The point of intersections between the line and the circle is:
\[
x = \frac{\pm \sqrt{\alpha^2 - \beta^2 + 1} - \alpha \beta}{\alpha^2 + 1}, \quad y = \frac{\beta \pm \alpha \sqrt{\alpha^2 - \beta^2 + 1}}{\alpha^2 + 1}
\]
We know that one point of intersection is $(1, 0)$. Given that $\alpha$ and $\beta$ can both be realized by $O(\wordsize)$ bits, this implies that the values $\sqrt{\alpha^2 - \beta^2 + 1}$ and $\sqrt{\alpha^2 - \beta^2 + 1}$ can also be realized by $O(\wordsize)$ bits.
\end{proof}

\newpage

\bibliographystyle{plain}
\bibliography{references}

\end{document}